\documentclass[11pt,english]{article}
\usepackage{amsmath,amssymb,amsthm}
\usepackage{multicol}
\usepackage[margin=1in]{geometry}
\usepackage{graphicx,color}
\usepackage{enumitem}
\usepackage{fullpage}
\usepackage[noblocks]{authblk}
\usepackage{tcolorbox}
\usepackage{babel}
\usepackage{wrapfig}
\usepackage{MnSymbol}
\usepackage{mdframed}
\usepackage{mathtools}
\usepackage{floatrow}
\usepackage{multirow}

\usepackage[ruled,linesnumbered,vlined]{algorithm2e}
\usepackage{tablefootnote}
\usepackage{thm-restate}
\usepackage{bbding}
\usepackage{pifont}
\usepackage{nicefrac}
\renewcommand{\arraystretch}{1.2}

\definecolor{Darkblue}{rgb}{0,0,0.4}
\definecolor{Brown}{cmyk}{0,0.61,1.,0.60}
\definecolor{Purple}{cmyk}{0.45,0.86,0,0}
\definecolor{Darkgreen}{rgb}{0.133,0.543,0.133}

\usepackage[colorlinks,linkcolor=Darkblue,filecolor=blue,citecolor=blue,urlcolor=Darkblue,pagebackref]{hyperref}
\usepackage[nameinlink]{cleveref}

\usepackage[colorinlistoftodos,prependcaption,textsize=tiny]{todonotes}

\newif\ifdraft 
\draftfalse

\newcommand{\namedref}[2]{\hyperref[#2]{#1~\ref*{#2}}}
\newcommand{\propref}[1]{\hyperref[#1]{property~(\ref*{#1})}}

\newtheorem{theorem}{Theorem}
\newtheorem{lemma}{Lemma}
\newtheorem{definition}{Definition}
\newtheorem{claim}{Claim}

\newtheorem*{problem*}{Problem}
\newtheorem*{question*}{Question}

\newcommand{\eps}{\varepsilon}
\renewcommand{\epsilon}{\varepsilon}
\newcommand{\diam}{\mathrm{diam}}

\newcommand{\dir}{\mathrm{dir}}
\newcommand{\lca}{\mathrm{lca}}
\newcommand{\MST}{\mathrm{MST}}
\newcommand{\opt}{\mathrm{OPT}}
\newcommand{\alg}{\mathrm{ALG}}
\newcommand{\N}{\mathbb{N}}
\newcommand{\R}{\mathbb{R}}

\newcommand{\old}[1]{{}}

\title{Online Spanners in Metric Spaces}
\date{}

\author{Sujoy Bhore\thanks{Indian Institute of Science Education and Research, Bhopal, India. Email: sujoy.bhore@gmail.com}
\quad
Arnold Filtser\thanks{Bar Ilan University, Ramat Gan, Israel. Email: arnold273@gmail.com}
\quad
Hadi Khodabandeh\thanks{University of California, Irvine, CA, USA. Email: khodabah@uci.edu}
\quad
Csaba D. T\'oth\thanks{California State University Northridge, Los Angeles, CA; and Tufts University, Medford, MA, USA. Email: csaba.toth@csun.edu}
}

\old{
\author[1]{Sujoy Bhore\thanks{sujoy.bhore@gmail.com}}
\affil[1]{Indian Institute of Science Education and Research, Bhopal, India}

\author[2]{Arnold Filtser\thanks{arnold273@gmail.com}}
\affil[2]{Bar Ilan University, Ramat Gan, Israel}

\author[3]{Hadi Khodabandeh\thanks{khodabah@uci.edu}}
\affil[3]{University of California, Irvine, CA, USA}

\author[4,5]{Csaba D. T\'oth\thanks{csaba.toth@csun.edu}}
\affil[4]{California State University Northridge, Los Angeles, CA, USA}
\affil[5]{Tufts University, Medford, MA, USA}
{}
}

\begin{document}
	
	\maketitle
	\vspace{-\baselineskip}
	
	\begin{abstract}
	
	Given a metric space $\mathcal{M}=(X,\delta)$, a weighted graph $G$ over $X$ is a metric $t$-spanner of $\mathcal{M}$ if for every $u,v \in X$, 
	$\delta(u,v)\le d_G(u,v)\le t\cdot \delta(u,v)$, where $d_G$ is the shortest path metric in $G$.
	In this paper, we construct spanners for finite sets in metric spaces in the online setting.
	Here, we are given a sequence of points $(s_1, \ldots, s_n)$, where the points are presented one at a time (i.e., after $i$ steps, we saw $S_i = \{s_1, \ldots , s_i\}$).
	The algorithm is allowed to add edges to the spanner when a new point arrives, however, it is not allowed to remove any edge from the spanner.
	The goal is to maintain a $t$-spanner $G_i$ for $S_i$ for all $i$, while minimizing the number of edges, and their total weight.

	We construct online $(1+\epsilon)$-spanners 
	in Euclidean $d$-space, 
	$(2k-1)(1+\eps)$-spanners for general metrics, and $(2+\eps)$-spanners for ultrametrics.
	Most notably, in Euclidean plane, we construct a $(1+\epsilon)$-spanner with competitive ratio $O(\eps^{-3/2}\log\eps^{-1}\log n)$, bypassing the classic  lower bound $\Omega(\eps^{-2})$ for lightness, which compares the weight of the spanner, to that of the MST.
\end{abstract}
	\pagenumbering{arabic}

	\section{Introduction}
	
	Let $\mathcal{M}=(P,\delta)$ be a finite metric space. Let $G=(P,E)$ be a graph on the points of $P$ in $\mathcal{M}$, where the edges are weighted with the distances between their endpoints. The graph $G$ is a \emph{$t$-spanner}, for $t\ge 1$, if $\delta_{G}(u,v) \le t\cdot \delta(u,v)$ for all $u,v\in P$, where $\delta_{G}(u,v)$ is 
the length of the shortest path between $u$ and $v$ in $G$, and $\delta(u,v)$ is the distance between $u$ and $v$ in $\mathcal{M}$. 
\footnote{Often in the literature, the input metric is the shortest path metric of a graph $G=(V,E,w)$, and a spanner is required to be a subgraph of the input graph (see e.g. \cite{althofer1993sparse}). Here we study metric spanners where there is no such requirement.}
The \emph{stretch factor} $t$ of $G$ is the maximum distortion between the metrics $\delta$ and $\delta_G$.
Spanners were first introduced by Peleg and Sch\"{a}ffer~\cite{peleg1989graph},
and since then they have turned out to be one of the fundamental graph
structures with numerous applications in the area of distributed systems and communication, distributed queuing protocol, compact routing schemes, etc.~\cite{demmer1998arrow,
	herlihy2001competitive,
	PelegU89a, 
	peleg1989trade}. 

The study of Euclidean spanners, where $P\subset \mathbb{R}^d$ with $L_2$-norm, was initiated by Chew~\cite{Chew89}. 
Since then a large body of research has been devoted to Euclidean spanners due to its vast range of applications across domains, such as topology control in wireless networks, efficient regression in metric spaces, approximate distance oracles, data structures, and many more~\cite{gottlieb2017efficient, gudmundsson2008approximate, schindelhauer2007geometric, Yao82}. Some of the results generalize to metric spaces with constant doubling dimensions~\cite{borradaile2019greedy} (the doubling dimension of $\mathbb{R}^d$ is $d$).

\smallskip\noindent \textbf{Lightness} and \textbf{sparsity} are two fundamental parameters for spanners. The lightness of a spanner $G=(P,E)$ is the ratio $w(G)/w(MST)$ between the total weight of $G$ and the weight of a minimum spanning tree (MST) on $P$.
The sparsity of $G$ is the ratio $|E(G)|/|E(MST)| \approx |E(G)|/|P|$ between the number of edges of $H$ and an MST. Since every spanner is connected and thus contain a spanning tree, the lightness and sparsity of a spanner $G$, resp., are trivial lower bounds for the ratio of $w(G)$ and $|E(G)|$ to the optimum weight and the number of edges.

\paragraph{Online Spanners.} 
We are given a sequence of points $(s_1,\ldots ,s_n)$ in a metric space, where the points are presented one-by-one, i.e., point $s_i$ is revealed at the step~$i$, and $S_i=\{s_1,\ldots,  s_i\}$ for $i\in\{1,\ldots , n\}$. The objective of an online algorithm is to maintain a $t$-spanner $G_i$ for $S_i$ for all $i$. The algorithm is allowed to \emph{add} edges to the spanner when a new point arrives, however it is not allowed to \emph{remove} any edge from the spanner. Moreover, the algorithm does not know the value of the total number points in advance. 

The performance of an online algorithm $\alg$ is measured by comparing it to the offline optimum $\opt$ using the standard notion of competitive ratio~\cite[Ch.~1]{BY98}. The \emph{competitive ratio} of an online $t$-spanner algorithm $\alg$ is defined as $\sup_\sigma \frac{\alg(\sigma)}{\opt(\sigma)}$, where the supremum is taken over all input sequences $\sigma$, $\opt(\sigma)$ is the minimum weight of a $t$-spanner for the (unordered) set of points in $\sigma$, and $\alg(\sigma)$ denotes the weight of the $t$-spanner produced by $\alg$ for this input sequence.
Note that, in order to measure the competitive ratio it is important that $\sigma$ is a finite sequence of points.

In the online minimum spanning tree problem,  points of a finite metric space arrive one-by-one, and we need to connect each new point to a previous point to maintain a spanning tree. Imase and Waxman~\cite{imase1991dynamic} proved $\Theta(\log n)$-competitiveness, which is the best possible bound. 
Later, Alon and Azar~\cite{AlonA93} studied this problem for points in Euclidean plane, and proved a lower bound $\Omega(\log n/\log \log n)$ for the competitive ratio. Their result was the first to analyze the impact of auxiliary points (Steiner points) on a geometric network problem in the online setting. Several algorithms were proposed over the years for the online minimum Steiner tree and Steiner forest problems, on graphs in both weighted and unweighted settings; see~\cite{alon2006general, awerbuch2004line, berman1997line, HajiaghayiLP17, naor2011online}. However, these algorithms do not provide any guarantee on the stretch factor. This leads to the following open problem. 

\begin{problem*}
	Determine bounds on the competitive ratios for the weight and the number of edges of online $t$-spanners, for $t\ge 1$.
\end{problem*}

Previously, Gupta et al.~\cite[Theorem~1.5]{GRTU17} constructed online spanners for terminal pairs in the same model we consider here. The analysis of~\cite{GRTU17} implicitly implies that, given a sequence of $n$ points in an online fashion in a general metric space, one can maintain a $O(\log n)$-spanner with $O(n)$ edges and $O(\log n)$ lightness, as pointed out by one of the authors~\cite{Umboh21}.
Recent work on online \emph{directed} spanners~\cite{GrigorescuLQ21} is not comparable to our results.


In the geometric setting, $(1+\eps)$-spanners are possible in any constant dimension $d\in \mathbb{N}$. Tight worst-case bounds $\Theta_d(\eps^{-d})$ and $\Theta_d(\eps^{1-d})$ on the lightness and sparsity of offline $(1+\eps)$-spanners have recently been established by Le and Solomon~\cite{le2019truly}. Online Euclidean spanners in $\mathbb{R}^d$ have been introduced by Bhore and T\'{o}th~\cite{BT-oes-21}.
In the real line (1D), they have established a tight bound of $O((\eps^{-1}/\log \eps^{-1}) \log n )$ for the competitive ratio of any online $(1+\eps)$-spanner algorithm for $n$ points. In dimensions $d\geq 2$, the dynamic algorithm 
\textsc{DefSpanner} of Gao et al.~\cite{gao2006deformable} maintains a $(1+\eps)$-spanner with $O_d(\eps^{-(d+1)}n)$ edges and  $O_d(\eps^{-(d+1)}\log n)$ lightness, and works under the online model (as it never deletes edges when new points arrive). However, no lower bound better than the 1-dimensional $\Omega((\eps^{-1}/\log \eps^{-1}) \log n )$ is currently known in higher dimensions.


\subsection{Our Contribution}

See \Cref{table:1} for an overview of our results.

\begin{table}[htbp]
	\begin{tabular}{|l|l|l|l|l|}
		\hline
		\textbf{Family}                                        & \textbf{Stretch}                & \textbf{\# of edges}                                            & \textbf{Lightness}                              & \textbf{Ref/comments}                                                                            \\ \hline
		General metrics & $(2k-1)(1+\eps)$            & $O(\eps^{-1}\log(\frac{1}{\eps})) n^{1+\frac{1}{k}}$ & $O(n^{\frac{1}{k}}\eps^{-1}\log^{2}n)$ & \Cref{thm:GreedyOnlineSpanner}                                               \\ \cline{2-5} 
		\multicolumn{1}{|c|}{}                                 & $O(\log n)$                     & $O(n)$                                          & $O(\log n)$                                  & \cite{GRTU17,Umboh21}                                                    \\ \hline
		$\alpha$-HST                                           & $2\,\frac{\alpha}{\alpha-1}$ & $n-1$                                                           & $1$                                             & \Cref{lem:alphaHstStretch,lem:GreedyHSTisMST} \\ \hline
		Ultrametric                           & $O(\eps^{-1})$         & $n-1$                                                           & $1+\eps$                                    & \Cref{thm:ultrametricStretchAlpha}                                           \\ \cline{2-5} 
		& $2+\eps$                    & $O(n\eps^{-1}\log \eps^{-1})$                  & $O(\eps^{-2})$                     & \Cref{thm:SpannerUltrametric}                                                \\ \hline
		Doubling  $d$-space   & $1+\eps$ & $\eps^{-O(d)}\, n$   & $\eps^{-O(d)}\,\log n$                                             & \textsc{DefSpanner}~\cite{gao2006deformable} \\ \hline
		Euclidean $d$-space & $1+\eps$         & $O_d(\eps^{-d})\, n$                                                           & $O_d(\eps^{-(d+1)}\log n)$                                    & \textsc{DefSpanner}~\cite{gao2006deformable}                                           \\ \cline{2-5} 
		& $1+\eps$                    & $O_d(\eps^{1-d})\, n$  & $\Omega(\eps^{-1}n)$ & ordered $\Theta$-graph~\cite{ruppert1991approximating}                                                \\ \cline{2-5} 
		& $1+\eps$                    & $\tilde{O}_d(\eps^{1-d})\, n$  & $O(\eps^{-d}\log n)$ & \Cref{thm:UB}                                                \\ \hline

		Real line & $1+\eps$ & $O(n)$   & $\tilde{\Theta}(\eps^{-1}\log n)$                                             &  ordered greedy \cite{BT-oes-21} \\ \hline\hline
		\textbf{Family}                                        & \textbf{Stretch}                & \textbf{\# of edges}                                            & \textbf{Competitive Ratio}                             & \textbf{Ref/comments}    
		\\ \hline
		
		General metrics
		& $2k-1$ & -   & $\Omega(\frac1k\cdot n^{\frac1k})$ & \Cref{thm:GeneralMetricLB} \\ \hline	
		
		Euclidean plane & $1+\eps$ & $\tilde{O}(\eps^{-1})\, n$   & $\tilde{O}(\eps^{-3/2}\log n)$                                             & \Cref{thm:UB2} \\ \hline
		
		$\mathbb{R}^d$ with $L_1$-norm 
		& $1+\eps$ & -   & $\Omega(\eps^{-d})$ & \Cref{thm:LB-L1} \\ \hline
		
	\end{tabular}
	\caption{Overview of online spanners algorithms. In the last three rows, we  compare the spanner weight directly with the optimum weight (rather than the MST) to bound the competitive ratio. \label{table:1}}
\end{table}

\paragraph{Upper Bounds for Points in $\mathbb{R}^d$.}
Under the $L_2$-norm in $\mathbb{R}^d$, 
for arbitrary constant $d\in \mathbb{N}$, we present an online algorithm for $(1+\eps)$-spanner with lightness $O_d(\eps^{-d} \log n)$ and sparsity $O(\eps^{1-d}\log \eps^{-1})$ (\Cref{thm:UB} in \Cref{ssec:first}). This improves upon the previous lightness bound of $O_d(\eps^{-(d+1)}\log n)$ by Gao et al.~\cite[Lemma~3.8]{gao2006deformable}. 
In the plane, we give a tighter analysis of the same algorithm and achieve an almost quadratic improvement of the \emph{competitive ratio} to $O(\eps^{-3/2}\log\eps^{-1}\log n)$ (\Cref{thm:UB2} in \Cref{ssec:second}). Recall that in the offline setting, $\Theta(\eps^{-2})$ is a tight worst-case bound for the lightness of a $(1+\eps)$-spanner in the plane~\cite{le2019truly}. We obtain a better dependence on $\eps$ by comparing the online spanner with an instance-optimal spanner directly, bypassing the comparison to an MST (i.e., lightness). The logarithmic dependence on $n$ cannot be eliminated in the online setting, based on the lower bound in $\R^1$~\cite{BT-oes-21}.

\paragraph{Lower Bounds for Points in $\mathbb{R}^d$.}
As a counterpart, we design a sequence of points that yields a $\Omega_d(\eps^{-d})$ lower bound for the competitive ratio for online  $(1+\eps)$-spanner algorithms in $\mathbb{R}^d$ under the $L_1$-norm (\Cref{thm:LB-L1} in \Cref{sec:LB}). This improves the previous bound of $\Omega(\eps^{-2}/\log \eps^{-1})$ in $\mathbb{R}^2$ under the $L_1$-norm. 
It remains open whether a similar lower bound holds in $\R^d$ under the $L_2$-norm; the current best lower bound is $\Omega((\eps^{-1}/ \log \eps^{-1})\log n )$, established in~\cite{BT-oes-21}, holds already for the real line ($d=1$).

\paragraph{Points in General Metrics.} 
In \Cref{sec:metric}, we study online spanners in general metrics. Note that it is not possible to obtain a spanner with stretch less than 3 with a subquadratic number of edges, even in the offline settings, for general metrics.
We analyze an online version of the celebrated greedy spanner algorithm, dubbed \emph{ordered greedy}. With stretch factor $t = (2k-1)(1+\eps)$ for $k\ge 2$ and $\eps\in(0,1)$, we show that it maintains a spanner with $O(\eps^{-1}\log\frac{1}{\eps}) \cdot n^{1+\frac{1}{k}}$ edges and $O(\eps^{-1}n^{\frac{1}{k}}\log^2 n)$ lightness for a sequence of $n$ points in a metric space (\Cref{thm:GreedyOnlineSpanner}).
We show (in \Cref{thm:GeneralMetricLB}) that these bounds cannot be significantly improved, by introducing an instance where every online algorithm will have $\Omega(\frac{1}{k}\cdot n^{1/k})$ competitive ratio on both sparsity and lightness.
Next, we establish the trade-off among stretch, number of edges and lightness for points in ultrametrics.
Specifically, we show that it is possible to maintain a $(2+\eps)$-spanner with $O(\eps^{-1}\log \eps^{-1})\cdot n$ edges and $O(\eps^{-2})$ lightness in ultrametrics (\Cref{thm:SpannerUltrametric}). 
Note that as the uniform metric (shortest path on a clique) is an ultrametric, any subquadratic spanner must have stretch at least $2$.

\subsection{Related Work}

\subsubsection{Dynamic \& Streaming Algorithms for Graph Spanners} 

A $t$-spanner in a graph $G=(V,E)$ is subgraph $H=(V,E')$ such that $\delta_H(u,v)\leq t\cdot \delta_G(u,v)$ for all pairs of vertices $u,v\in V$. 
That is, the stretch $t$ is the maximum distortion between the graph distances $\delta_G$ and $\delta_H$. Importantly, when $G$ changes (under edge/vertex insertions or deletions), the underlying metric $\delta_G$ changes, as well. The distance $\delta_G(u,v)$ may dramatically decrease upon the insertion of the edge $uv$. In contrast, our model assumes that the distances in the underlying metric space $\mathcal{M}=(P,\delta)$ remain fixed, but the algorithm can only see the distances between the points that have been presented. For this reason, our results are not directly comparable to models where the underling graph changes dynamically. 


For \emph{unweighted} graphs with $n$ vertices, the current best fully dynamic and single-pass streaming algorithms can maintain spanners that achieve almost the same stretch-sparsity trade-off available for the static case: $2k-1$ stretch and $O(n^{1+\frac{1}{k}})$ edges, for $k\geq 1$, which is attained by the greedy algorithm~\cite{althofer1993sparse}, and conjectured to be optimal due to the  Erd\H{o}s girth conjecture~\cite{Erdos64}.
In the dynamic model, the objective is design algorithms and data structures that minimize the worst-case update time needed to maintain a $t$-spanner for $S$ over all steps, regardless of its weight, sparsity, or lightness. See~\cite{BaswanaKS12, BergamaschiHGWW21, BernsteinFH19, BodwinK16} for some excellent work on dynamic spanners. 
%
In the streaming model the input is a sequence (or stream) of edges representing the edge set $E$ of the graph $G$. A (single-pass) streaming algorithm decides, for each newly arriving edge, whether to include it in the spanner. The graph $G$ is too large to fit in memory, and the objective is to optimize work space and update time~\cite{Baswana08,BeckerFKL21, Elkin11, FeigenbaumKMSZ08, FiltserKN21, McGregor14}. 
%


\subsubsection{Incremental Algorithms for Geometric Spanners}
\label{ssec:previous}

We briefly review three previously known incremental $(1+\eps)$-spanner algorithms in Euclidean $d$-space from the perspective of competitive analysis. 

\paragraph{Deformable Spanners.}
Gao et al.~\cite{gao2006deformable} designed a dynamic \textsc{DefSpanner} algorithm that maintains a $(1+\eps)$-spanner for a dynamic set $S$ in Euclidean $d$-space. For point insertions, it only adds new edges, so it is an online algorithm, as well.  It maintains a $(1+\eps)$-spanner with $O_d(\eps^{-d})\cdot n$ edges and  $O_d(\eps^{-(d+1)}\log n)$ lightness. Since the $\|\MST(S)|$ is a lower bound for the optimal spanner weight, its competitive ratio is also $O_d(\eps^{-(d+1)}\log n)$.
The key ingredient of \textsc{DefSpanner} is \emph{hierarchical nets}~\cite{Har-PeledM06,KrauthgamerL04,Roditty12}, a form of hierarchical  clustering, which can be maintained dynamically. Hierarchical nets naturally generalize to doubling spaces, and so \textsc{DefSpanner} also maintains a $(1+\eps)$-spanner with $\eps^{-O(d)}\cdot n$ edges and lightness $\eps^{-O(d)}$ in for doubling dimension~$d$~\cite{GottliebR08,Roditty12}.

\paragraph{Well-Separated Pair Decomposition (WSPD).}
Well-separated pair decomposition was introduced by Callahan and Kosaraju~\cite{CallahanK93}  (see also~\cite{GudmundssonK18,HarPeled11,narasimhan2007geometric,Smid18}). 
For a set $S$ in a metric space, a WSPD is a collection of unordered pairs  $W=\{\{A_i,B_i\}: i\in I\}$ such that 
(1) $A_i,B_i\subset S$ for all $i\in I$; 
(2) $\min\{\|ab\|: a\in A_i, b\in B_i\}\leq \varrho\cdot \max\{\diam(A_i),\diam(B_i)\}$ for all $i\in I$, where $\varrho$ is the \emph{separation ratio};  
(3) for each point pair $\{a,b\}\subset S$ there exists a pair $\{A_i,B_i\}$ such that $A_i$ and $B_i$ each contain one of $a$ and $b$.
Given a WSPD with separation ratio $\varrho>4$, any graph that contains at least one edge between $A_i$ and $B_i$, for all $i\in I$, is a spanner with stretch $t=1+8/(\varrho-4)$. Setting $\varrho\geq 12\eps^{-1}$ for $0<\eps<1$, we obtain  $t\leq 1+\eps$. 

Hierarchical clustering provides a WSPD~\cite[Ch.~3]{HarPeled11}. 
Perhaps the simplest hierarchical subdivisions in $\R^d$ are quadtrees.
Let $\mathcal{T}$ be a quadtree for a finite set $S\subset \R^d$. 
The root of $\mathcal{T}$ is an axis-aligned cube of side length $a_0$, which  contains $S$; it is recursively subdivided into $2^d$ congruent cubes until each leaf cube contains at most one point in $S$. For all pairs of cubes $\{Q_1,Q_2\}$ at level $\ell$ of $\mathcal{T}$, create a pair $\{A_i,B_i\}$ with $A_i=Q_1\cap S$ and $B_i=Q_2\cap S$ whenever $D_\ell\leq \mathrm{dist}(Q_1,Q_b)< 2D_\ell$ for $D_\ell=\varrho\cdot \diam(Q_1) = 12\eps^{-1}\cdot \sqrt{d}\cdot a_0/2^\ell$; and repeat for all levels $\ell\geq 0$. Properties (1)--(3) of a WSPD are easily verified~\cite[Ch.~3]{HarPeled11}. 
The resulting $(1+\eps)$-spanner has $O_d(\eps^{-d})\cdot n$ edges~\cite{HarPeled11,Har-PeledM06} and lightness $O_d(\eps^{-(d+1)}\log n)$~\cite{BT-oes-21}. 

For point insertions in $\R^d$, a dynamic quadtree only adds nodes, which in turn creates new pairs in the WSPD, and new edges in the spanner. This is an online algorithm with the same guarantees as \textsc{DefSpanner}~\cite{BT-oes-21,Har-PeledM06} (see also~\cite{FischerH05} for an efficient implementation).

\paragraph{Ordered Yao-Graphs and $\Theta$-Graphs.} 
One of the first constructions for (offline) sparse $(1+\eps)$-spanner in Euclidean $d$-space were the Yao- and $\Theta$-graphs~\cite{Clarkson87,keil1988approximating,ruppert1991approximating}. 
Incremental versions of Yao-graphs and $\Theta$-graphs were introduced by Bose et al.~\cite{bose2004ordered}. Let $S=\{s_1,\ldots , s_n\}$ be an ordered set of points in $\R^2$. For each $s_i\in S$, partition the plane into $k$ cones with apex $s$ and aperture $2\pi/k$. The \emph{ordered Yao-graph} $Y_k(S)$ contains an edge between $s_i$ and a closest \emph{previous} point in $\{s_j: j<i\}$ in each cone.
The graph $\Theta_k(S)$ is defined similarly, but in each cone the distance to the apex is measured by the orthogonal projection to a ray within the cone. 
Bose et al.~\cite{bose2004ordered} showed that the ordered Yao- and $\Theta$-graphs have spanning ratio at most $1/(1-2\sin(\pi/k))$ for $k>8$; tighter bounds were later obtained in~\cite{BoseCMRV16}. In particular, the ordered Yao- and $\Theta$-graphs are $(1+\eps)$-spanners for $k\geq \Omega(\eps^{-1})$. 

The construction generalizes to $\R^d$ for all $d\in \mathbb{N}$~\cite{ruppert1991approximating}. 
For an angle $\alpha\in (0,\pi)$, let $A\subset \mathbb{S}^{d-1}$ be a maximal set of points in the $(d-1)$-sphere such that $\min_{a,b\in A} \mathrm{dist}(a,b)\leq \alpha$ (in radians). A standard volume argument shows that $|A|\leq O_d(\alpha^{1-d})$. 
For each $a_i\in A$, create a cone $C_i$ with apex at the origin $o$, aperture $\alpha$, and symmetry axis $oa_i$. Note that $\R^d\subseteq \bigcup_{i} C_i$.
Given a finite set $P\subset \R^d$, we translate each cone $C_i$ to a cone $C_i(p)$ with apex $p\in P$. For every cone $C_i(p)$, the Yao-graph contains an edge between $p$ and a closest point in $P\cap C_i(p)$. 
For every $\eps>0$ and $d\in \N$, there exists an angle $\alpha=\alpha(d,\eps)=\Theta_d(\eps)$ for which the Yao-graph is a $(1+\eps)$-spanner for every finite set $P\subset \R^d$. 

Ordered Yao- and $\Theta$-graphs give online algorithms for maintaining a $(1+\eps)$-spanner for a sequence of points in $\R^d$. The sparsity of these spanners is bounded by the number of cones per vertex, $O_d(\eps^{1-d})$, which matches the (offline) lower bound of $\Omega_d(\eps^{1-d})$~\cite{le2019truly}. However, their weight may be significantly higher than optimal: For $n$ equally spaced points in a unit circle, in any order, Yao- and $\Theta$-graphs yield $(1+\eps)$-spanners of weight $\Omega(\eps^{-1}n)$, hence lightness $\Omega(\eps^{-1} n)$, while the optimum weight is $O(\eps^{-2})$~\cite{le2019truly}. 

\old{
	\smallskip\noindent\textbf{Incremental and Dynamic Spanners.}\todo{S: Related work became big. Please check if it is possible to squeeze.}
	An incremental algorithm is given a sequence of input, and finds a sequence of solutions that build incrementally while adapting to the changes in the input. Note that, an online algorithm also takes in a sequence of input and produce incremental solutions. However, an online algorithm does not know the input sequence in advance unlike their incremental counterpart. 
	Arya et al.~\cite{arya1994randomized} designed a randomized incremental algorithm for $n$ points in $\R^d$, where the points are inserted in a random order, and maintains a $t$-spanner of $O(n)$ size and $O(\log n)$ diameter. Moreover, their algorithm can also handle random insertions and deletions in $O(\log^d n \log \log n)$ expected amortized update time. Later, Bose et al.~\cite{bose2004ordered} presented an insertion only algorithm to maintain a $t$-spanner of $O(n)$ size and $O(\log n)$ diameter in $\R^d$. 
	
	In the dynamic model, the objective is design algorithms and data structures that minimize the worst-case update time needed to maintain a $t$-spanner for $S$ over all steps, regardless of its weight, sparsity, or lightness. Notice that dynamic algorithms are allowed to 
	add or delete edges in each step, while online algorithms cannot delete edges. However, if a dynamic (or dynamic insert-only) algorithm always adds edges for a sequence of points insertions, it is also an online algorithm, and one can analyze its competitive ratio. 
	Fischer and Har-Peled~\cite{FischerH05} used dynamic compressed quadtrees to maintain a WSPD-based $(1+\eps)$-spanner for $n$ points in $\R^d$ in expected $O([\log n+\log \eps^{-1}]\,\eps^{-d}\log n)$ update time. Their algorithm works under the online model, too, however, they have not analyzed the weight of the resulting spanner.
	Gao et al.~\cite{gao2006deformable} used hierarchical clustering for dynamic spanners in $\R^d$. Their {\sc DefSpanner} algorithm is fully dynamic with $O(\log \Delta)$ update time, where $\Delta$ is the spread\footnote{The \emph{spread} of a finite set $S$ in a metric space is the ratio of the maximum pairwise distance to the minimum pairwise distance of points in $S$; and $\log \Delta\geq \Omega(\log n)$ in doubling dimensions.} of the set $S$. 
	They maintain a $(1+\eps)$-spanner of weight $O(\eps^{-(d+1)}\|MST(S)\|\log \Delta)$, and for a sequence of point insertions, {\sc DefSpanner} only adds edges. As $\opt\geq \|MST(S)\|$,  {\sc DefSpanner} can serve as an online algorithm with competitive ratio $O(\eps^{-(d+1)}\log \Delta)$. 
	
	Gottlieb and Roditty~\cite{gottlieb2008optimal} studied dynamic spanners in more general settings. For every set of $n$ points in a metric space of bounded doubling dimension\footnote{A metric is said to be of a \emph{constant doubling dimension} if a ball with radius $r$ can be covered by at most a constant number of balls of radius $r/2$.}, they constructed a $(1+\eps)$-spanner whose maximum degree is
	$O(1)$ and that can be maintained under insertions and deletions in $O(\log n)$ amortized update time per operation. 
	Later, Roditty~\cite{Roditty12} designed fully dynamic geometric $t$-spanners with optimal $O(\log n)$ update time for $n$ points in $\mathbb{R}^d$. Recently, Chan et al.~\cite{ChanHJ20} introduced \emph{locality sensitive orderings} in $\R^d$, which has applications in several proximity problems, including spanners. They obtained a fully dynamic data structure for maintaining a 
	$(1 + \eps)$-spanners in Euclidean space with logarithmic update time and linearly many edges. However, the spanner weight has not been analyzed for any of these constructions. Dynamic spanners have been subject to investigation in abstract graphs, as well. See~\cite{BaswanaKS12, BergamaschiHGWW21, BernsteinFH19} for some recent progress on dynamic graph spanners.  
}

\paragraph{Online Steiner Spanners.}
An important variant of online spanners is when it is allowed to use auxiliary points (Steiner points) which are not part of the input sequence of points, but are present in the metric space. An online algorithm is allowed \emph{add} Steiner points, 
however, the spanner must achieve the given stretch factor only for the input point pairs.
It has been observed through a series of work in recent years, that Steiner points allow for substantial improvements over the bounds on the sparsity and lightness of Euclidean spanners in the offline settings and highly nontrivial insights are required to argue the bounds for Steiner spanners, and often they tend to be even more intricate than their non-Steiner counterpart; see~\cite{BT-lessp-21, BT-oess-21, le2019truly, le2020light}. 
Bhore and T\'{o}th~\cite{BT-oes-21} showed that if an algorithm can use Steiner points, then the competitive ratio for weight improves to $O(\eps^{(1-d)/2} \log n)$ in the Euclidean $d$-space. 



	\section{Upper Bounds in Euclidean Spaces}
	\label{sec:Euclid}
	
	We present an online algorithm for a sequence of $n$ points in Euclidean $d$-space (\Cref{ssec:first}). It combines features from several previous approaches, and maintains a $(1+\eps)$-spanner of lightness $O_d(\eps^{-d}\log n)$ and sparsity $O_d(\eps^{1-d}\log \eps^{-1})$ for $d\geq 1$. Lightness is an upper bound for the competitive ratio for weight; the sparsity almost matching the optimal bound $O_d(\eps^{1-d})$ attained by ordered Yao-graphs.
	In the plane ($d=2$), we show that the same algorithm achieves competitive ratio $O(\eps^{-3/2}\log \eps^{-1}\log n)$ using a tighter analysis: A charging scheme that charges the weight of the online spanner to a minimum weight spanner  (\Cref{ssec:second}).

	\subsection{An Improvement in All Dimensions}
	\label{ssec:first}
	
	We combine features from two incremental algorithms for geometric spanners, and obtain an online $(1+\eps)$-spanner algorithm for a sequence of $n$ points in $\R^d$. We maintain a dynamic quadtree for hierarchical clustering, and use a modified ordered Yao-graph in each level of the hierarchy. In particular, we limit the weight of the edges in the Yao-graph in each level of the hierarchy (thereby avoiding heavy edges). We start with an easy observation.
	
	\begin{lemma}\label{lem:restrict}
		Let  $G=(S,E)$ be a $t$-spanner and let $w>0$. Let $G'=(S,E')$, where 
		$E'=\{e\in E: \|e\|\leq w\}$ is the set of edges of weight at most $w$.
		Then for every $a,b\in S$ with $\|ab\|<w/t$, graph $G'$ contains 
		an $ab$-path of weights at most $t\, \|ab\|$.
	\end{lemma}
	\begin{proof}
		Since $G$ is a $t$-spanner, it contains an $ab$-path $P_{ab}$ of weight at most $t\, \|ab\|\leq w$. By the triangle inequality, every edge in this path has weight at most $w$, hence present in $G'$. Consequently $G'$ contains $P_{ab}$.
	\end{proof}
	
	\paragraph{Online Algorithm $\alg_1$.}
	The input is a sequence of points $(s_1,s_2,\ldots )$ in $\R^d$, $d\geq 1$. The set of the first $n$ points is denoted by $S_n=\{s_i: 1\leq i\leq n\}$. For every $n$, we dynamically maintain a quadtree $\mathcal{T}_n$ for $S_n$. Every node of $\mathcal{T}_n$ corresponds to a cube.
	The root of $\mathcal{T}_n$, at level $0$, corresponds to a cube $Q_0$ of side length $a_0=\Theta(\diam(S_n))$. At every level $\ell\geq 0$, there are at most $2^{d\ell}$ interior-disjoint cubes, each of side length $a_\ell=a_0\, 2^{-\ell}$. A cube $Q\in \mathcal{T}_n$ is \emph{nonempty} if $Q\cap S_n\neq \emptyset$. For every nonempty cube $Q$, we maintain a representative $s(Q)\in Q\cap S_n$, selected at the time when $Q$ becomes nonempty. At each level $\ell$, let $P_\ell$ be the sequence of representatives, in the order in which they are created.  
	
	For each level $\ell$, we maintain a modified ordered Yao-graph $G_\ell=(P_\ell,E_\ell)$ as follows. When a new point $p$ is inserted into $P_\ell$, cover $\R^d$ with $\Theta_d(\eps^{1-d})$ cones of aperture $\alpha(d,\eps)$ as in the construction of Yao-graphs. In each cone $C_i$, find a point $q_i\in C_i\cap P_\ell$ 
	closest to $p$; and add $pq_i$ to $E_\ell$ if $\|pq_i\|< 24a_\ell\sqrt{d}\cdot \eps^{-1}$.
	The algorithm maintains the spanner $G=\bigcup_{\ell\geq 0} G_\ell$. 
	
	\begin{theorem}\label{thm:UB}
		Let $d\geq 1$ and $\eps\in (0,1)$. The online algorithm $\alg_1$ maintains, for a sequence of $n$ points in Euclidean $d$-space, an $(1+O(\eps))$-spanner with weight $O_d(\eps^{-d}\log n)\cdot \|MST\|$ and $O_d(\eps^{1-d} \log\eps^{-1})\cdot n$ edges.
	\end{theorem}
	Note that \Cref{thm:UB} implies that the competitive ratio of this algorithm is also $O_d(\eps^{-d}\log n)$.
	\begin{proof}
		\textbf{Stretch Analysis.} 
		We give a bound on the stretch factor in two steps: First, we define an auxiliary graph $H=(S,E')$ which is a $(1+\eps)$-spanner for $S$ by the analysis of WSPDs. Then we show that $G$ contains an $ab$-path of weight at most $(1+\eps)\|ab\|$ for each edge of $H$. Overall, the stretch of $G$ is at most $(1+\eps)^2=(1+O(\eps))$ for all $a,b\in S$. 
		
		\noindent\emph{First Layer: WSPD.} 
		For each level $\ell\geq 0$, let $H_\ell=(P_\ell,E_\ell')$ be the graph that contains an edge between two representatives $a,b\in P_\ell$ whenever $\|ab\|\leq 12a_\ell\sqrt{d}\cdot \eps^{-1}$. Let $H=\bigcup_{\ell\geq 0} H_\ell$. 
		The auxiliary graph $H_\ell$ contains an edge between the representatives of any such pair of cubes at level $\ell$. 
		As noted \Cref{ssec:previous}, $H=\bigcup_{\ell\geq 0}H_\ell$ is a $(1+\eps)$-spanner (cf.~\cite{HarPeled11,Har-PeledM06}). 
		
		\noindent\emph{Second Layer: Near-Sighted Yao-graphs.} 
		As $H$ is a $(1+\eps)$-spanner, for every $a,b\in S_n$, it contains an $ab$-path of weight at most $(1+\eps)\|ab\|$. Consider such a path $P_{ab}=(a=p_0,\ldots , p_m=b)$. Each edge $p_{i-1}p_i$ is in $H_\ell$ for some $\ell\geq 0$.
		By construction, every edge in $H_\ell$ has weight at most $12a_\ell\sqrt{d}\cdot \eps^{-1}$.
		For every level $\ell$, the ordered Yao-graph $Y(P_\ell)$ with angle $\alpha(d,\eps)$ is a $(1+\eps)$-spanner. The graph $G_\ell=(P_\ell,E_\ell)$ constructed by $\alg_1$ at level $\ell$ is a subgraph of $Y(P_\ell)$. 
		By \Cref{lem:restrict}, for every $p,q\in P_\ell$ with $\|pq\|\leq 12a_\ell\sqrt{d}\cdot \eps^{-1}$, graph $G_\ell$ 
		contains a $pq$-path of weight at most $(1+\eps)\|pq\|$. 
		
		Overall, $H$ contains an $ab$-path $P_{ab}=(p_0,\ldots , p_m)$ of weight at most $(1+\eps)\|ab\|$. For each edge $p_{i-1}p_i$ of $P_{ab}$, graph $G$ contains a $p_{i-1}p_i$-path of weight $(1+\eps)\|p_{i-1}p_i\|$. The concatenation of these paths is an $ab$-path of weight $(1+\eps)^2\|ab\|\leq (1+O(\eps))\|ab\|$.
		
		\paragraph{Weight Analysis.}
		We may assume w.l.o.g.\ that the root of the quadtree $\mathcal{T}_n$ is the unit cube $[0,1]^d\subset \R^d$, which has diameter $\sqrt{d}$. This implies $\diam(S_n)\leq \sqrt{d}=O_d(1)$. Assume further that $n>1$, and $\frac14\leq \diam(S_n)\leq \|MST(S_n)\|$. 
		
		Every edge in $E_\ell$ at level $\ell$ has weight $O_d(\eps^{-1}\, 2^{-\ell})$.
		In particular, every edge at level $\ell\geq 2\log n$ has weight  $O_d(\eps^{-1}/n^2)$; and the total weight of these edges is $O_d(\eps^{-1})\leq O_d(\eps^{-1}\|MST(S_n)\|)$. 
		
		It remains to bound the weight of the edges on levels $\ell=1,\ldots,\lfloor 2\log n\rfloor$. At level $\ell$ of the quadtree $\mathcal{T}_n$, there are at most $2^{d\ell}$ nodes, hence $|P_\ell|\leq 2^{d\ell}$. If $|P_\ell|< 3^d$, then $G_\ell$ has at most $O(3^{2d})=O_d(1)$ edges, each of weight at most $\diam(P_\ell)\leq\diam(S_n)\leq \|\MST(S_n)\|$, and so $\|E_\ell\|\leq O_d(\|\MST(S_n)\|)$. 
		Assume now that $|G_\ell|\geq 3^d$. By the definition of ordered Yao-graphs, each vertex inserted into $P_\ell$ adds $\Theta(\eps^{1-d})$ new edges, each of weight $O(\eps^{-1}\, 2^{-\ell})$. The total weight of the edges in $G_\ell$ is at most
		\begin{equation}\label{eq:1}
			\|E_\ell\|
			\leq |P_\ell|\cdot \eps^{1-d} \cdot \max_{e\in E_\ell} \|e\| 
			\leq O_d( |P_\ell| \,\eps^{-d}\, 2^{-\ell}). 
		\end{equation}
		
		We next derive a lower bound for $\|\MST(S_n)\|$ in terms of $|P_\ell|$, when $|P_\ell|>1$ and $\ell>2$, using a standard volume argument. Define a graph on the vertex set $P_\ell$ such that two nodes $p,q\in P_{\ell}$ are adjacent iff $p$ and $q$ lie in neighboring quadtree cells of level $\ell$. Since every quadtree cell has $3^d-1$ neighbors, this graph is $(3^d-1)$-degenerate, and contains an independent set $I_\ell$ of size at least $(3^d-1)^{-1}|P_\ell| = \Omega_d(|P_\ell|)$. The distance between any two disjoint quadtreee cells at level $\ell$ is at least $2^{-\ell}$. Consequently, the open balls of radius $2^{-(\ell+1)}$ centered at the points in $I_{\ell}$ are pairwise disjoint. None of the balls contains $S_n$ for $\ell>2$, as the diameter of each of ball is $2^{-\ell}$ while $\diam(S_n)\geq \frac14$. For all $\ell>2$, $\MST(S_n)$ contains the center of each ball and a point in its exterior; hence the intersection of $\MST(S_n)$ and each ball contains a path from the center to a boundary point, which has weight at least $2^{-(\ell+1)}$. Summation over $|I_\ell|$ disjoint balls yields
		\begin{equation}\label{eq:2}
			\|MST(S_n)\|\geq |I_\ell|\cdot 2^{-(\ell+1)} \geq \Omega_d( |P_\ell|\, 2^{-\ell}).
		\end{equation}
		Comparing inequalities~\eqref{eq:1} and \eqref{eq:2}, we obtain 
		$\|E_\ell\|\leq O_d(\eps^{-d})\cdot \|MST(S_n)\|$.
		Summation over all levels $\ell\in \mathbb{N}$ yields 
		$\|E\|\leq O_d(\eps^{-d} \log n)\cdot \|\MST(S_n)\|$, as claimed.
		
		\paragraph{Sparsity Analysis.}
		We show that $G$ has $O_d(\eps^{1-d}\log \eps^{-1})\cdot n$ edges. Har-Peled proved that the auxiliary graph $H$ is $O(\eps^{-d})$-degenerate, and so it has $O_d(\eps^{-d})\cdot n$ edges~\cite[Lemma~3.9]{HarPeled11,Har-PeledM06}. As $G$ is a subgraph of $H$, hence has $O_d(\eps^{-d})\cdot n$
		edges as well. We improve this bound using a charging scheme.
		
		For the quadtree $\mathcal{T}_n$ maintained by algorithm~$\alg_1$, let $\mathcal{T}'_n$ denote the \emph{compressed quadtree}, which is obtained from $\mathcal{T}_n$ by removing all leaves that correspond to empty cubes, and supressing nodes with a single child~\cite{BergCKO08,HarPeled11}. For $n$ points in $\R^d$, the compressed quadtree has $O_d(n)$ nodes
		(which are nodes of the original quadtree, as well). For each node $Q$ of $\mathcal{T}_n'$, algorithm $\alg_1$ adds
		$O_d(\eps^{1-d})$ edges between the representative $s(Q)$ and the closest points in each cone $C_i$ (in $P_\ell$, where $\ell\geq 0$ is the level of $Q$ in $\mathcal{T}_n$).
		The total number of these edges for all nodes of $\mathcal{T}_n'$ is $O(\eps^{1-d})\cdot n$.
		
		It remains to consider the nodes of the quadtree $\mathcal{T}_n$ that are compressed in $\mathcal{T}_n'$. Every compressed node
		is part of a descending chain of single-child nodes in $|mathcal{T}_n$. The number of such chains is $O_d(n)$, as each chain has a unique direct descendant in $\mathcal{T}_n'$.
		Let $Q_k,\ldots ,Q_\ell$ be a maximal chain of single-child nodes in $\mathcal{T}_n$, where $Q_j$ is on level $j$ of $\mathcal{T}_n$ for $j=k,\ldots , \ell$. These are nested cubes $Q_k\subset Q_{k-1}\subset \ldots \subset Q_1$ with a common representative, $s=q(Q_k)=\ldots =s(Q_\ell)$; see \Cref{fig:sparsity}. Let $C_i$ be a cones with apex $s$ and aperture $\alpha(d,\eps)$ in algorithm $\alg_1$; and let $q_{i,j}$ denote the closest point to $s$ in $C_i\cap P_j$ for $j=k,\ldots , \ell$. If a point $q_{i,j}\in P_j$ represents some compressed cube $Q'$ (in another compressed chain), then $q_{i,j}$ represents the parent of $Q'$, as well. In  this case, $q_{i,j}\in P_{\ell-1}$, which implies $q_{i,j}=q_{i,j-1}$. Consequently, $q_{i,j}=q_{i,j-1}=\ldots =q_{i,k}$. We may assume that only $q_{i,k}$ represents a compressed node.

		\begin{figure}[htbp]
			\centering
			\includegraphics[width=0.7\textwidth]{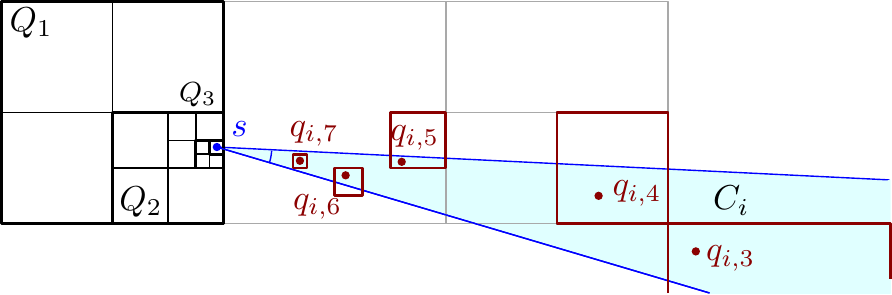}
			\caption{A point $s$ is the representative of five nested squares in the quadtree.  The closest point to $s$ is $q_{i,\ell}\in C_i\cap P_\ell$ in the cone $C_i$ at level $\ell=3,\ldots , 7$.}
			\label{fig:sparsity}
		\end{figure}

		The first $\leq \lceil \log \eps^{-1}\rceil$ nodes (i.e., $Q_j$ for $k\leq +\lceil \log \eps^{-1}\rceil$) jointly contribute $O(\eps^{1-d}\log \eps^{-1})$ edges to $G$. Summation over all compressed chains yields $O(\eps^{-1}\log \eps^{-1})\cdot n$ edges.
		For the remaining nodes in the chain (that is, nodes $Q_j$ for $k<j\leq \ell-\log \eps^{-1}$), we use thefollowing charging scheme: Charge the edge $sq_{i,j}$ to $q_{i,j}$. Since $j\neq k$, then $q_{i,j}$ represents a noncompressed node at level $j$ 
		of $\mathcal{T}_n$. Next we bound the charges received by $q_{i,j}$.
		
		We claim that for every noncompressed node $Q$,
		the representative $q=s(Q)$ receives at most $O_d(1)$ units of charges.
		Indeed, suppose that $q\in P_\ell$ and an edge $sq$ has been charged to $q$. Then $\|sq\|\leq 24a_\ell \sqrt{d}\cdot \eps^{-1}$. However, $s$ is the only point in the cube $Q'_s:=Q_{j-\lceil \log \eps^{-1}\rceil}$ of side length $a_\ell\cdot 2^{\lceil \log \eps^{-1}\rceil}\geq a_\ell\cdot \eps^{-1}$  and $\diam(Q'_s)\geq a_\ell\sqrt{d}\cdot \eps^{-1}$.
		Consequently, $Q_s'$ lies in the ball $B_q$ of radius $25a_\ell \sqrt{d}\cdot \eps^{-1}$ centered at $q$.
		However, comparing the volumes of $B_q$ and $Q'_s$ shows that $B_q$ contains $O_d(1)$ interior-disjoint cubes $Q_s'$,
		and so $q$ is charged at most $O_d(1)$ times.
		Summation over all all $O_d(n)$ noncompressed nodes over all levels $\ell>0$ shows that
		the total number of edges that participate in the charging scheme is $O_d(n)$.
		
		Overall, we have shown that $G$ has at most $O(\eps^{1-d}\log \eps^{-1})\cdot n$ edges.
	\end{proof}

	\subsection{Further Improvements in the Plane}
	\label{ssec:second}
	
	We presents a tighter analysis of algorithm $\alg_1$ for $d=2$ that compares the spanner weight to the offline optimum weight, and bypasses the comparison with the MST (i.e., lightness). 
	
	\paragraph{Minimum-Weight Euclidean $(1+\eps)$-Spanner.}
	For any $a,b\in \mathbb{R}^d$, an $ab$-path $P_{ab}$ of Euclidean weight at most $(1+\eps)\|ab\|$ lies in the ellipsoid $\mathcal{E}_{ab}$ with foci $a$ and $b$ and great axes of weight $(1+\eps)\|ab\|$; see \Cref{fig:ellipse}.
	A key observation is that the minor axis of $\mathcal{E}_{ab}$ is $((1+\eps)^2-1^2)^{1/2}\, \|ab\| \approx \sqrt{2\eps}\,\|ab\|$. 
	Furthermore, Bhore and T\'oth~\cite{BT-oess-21} recently observed that the directions of ``most'' edges of the path $P_{ab}$ are ``close'' to the direction of $ab$. Specifically, if we denote by $E(\alpha)$ the set of edges $e$ in $P_{ab}$ with $\angle(ab,e)\leq\alpha$, then the following holds.
	
	\begin{lemma}[Bhore and T\'oth~\cite{BT-oess-21}] \label{lem:old}
		Let $a,b\in \mathbb{R}^d$ and let $P_{ab}$ be an $ab$-path of weight $\|P_{ab}\|\leq (1+\eps)\|ab\|$. Then for every $i\in \{1,\ldots, \lfloor1/\sqrt{\eps}\rfloor\}$,
		we have $\|E(i\cdot \sqrt{\eps})\|\geq (1-2/i^2)\,\|ab\|$.
	\end{lemma}

	\begin{figure}[htbp]
		\centering
		\includegraphics[width=0.7\textwidth]{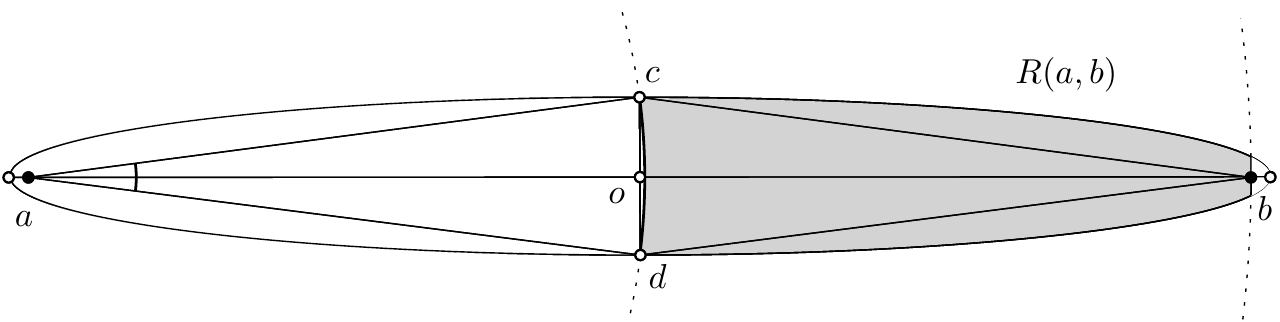}
		\caption{Any $ab$-path of weight at most $(1+\eps)\|ab\|$ lies in the ellipse $\mathcal{E}_{ab}$ with foci $a$ and $b$. The shaded region $R(a,b)$ is the part of the ellipse $\mathcal{E}_{ab}$ between two concentric circles centered at $a$.}
		\label{fig:ellipse}
	\end{figure}
	
	Let $R(a,b)=\mathcal{E}_{ab} \cap \mathcal{N}(a,b)$, where $\mathcal{N}(a,b)$ is the annulus 
	bounded by two concentric spheres centered at $a$, of radii $\frac{1+\eps}{2}\,\|ab\|$ and $\|ab\|$; see \Cref{fig:ellipse} for an example. 
	\begin{lemma}\label{lem:ellipse}
		If $0<\eps<\frac19$, then every $ab$-path $P_{ab}$ of weight at most $\|P_{ab}\|\leq (1+\eps)\|ab\|$ contains interior-disjoint line segments $s\subset R(a,b)$ of total weight at least $\frac{1}{3}\, \|ab\|$ such that $\angle(\overrightarrow{ab},s)\leq 3\cdot \sqrt{\eps}$. 
	\end{lemma}
	\begin{proof}
		Since the distance between the two concentric circles is $\frac{1-\eps}{2}\,\|ab\|$,
		every $ab$-path contains a subpath of weight at least $\frac{1-\eps}{2}\,\|ab\|$
		in the ans $\mathcal{N}(a,b)$. 
		
		Let $P_{ab}$ be an $ab$-path of weight at most $(1+\eps)\|ab\|$. As noted above $P_{ab}\subset \mathcal{E}_{ab}$. Hence, $\|P_{ab}\cap \mathcal{N}(a,b)\| = \|P_{ab}\cap R(a,b)\| \geq \frac{1-\eps}{2}\,\|ab\|$ in $R(ab)$; and so 
		$\|P_{ab}\setminus R(a,b)\| 
		= \|P_{ab}\| - \|P_{ab} \cap R(a,b)\| 
		\leq \frac{1+3\eps}{2}\, \|ab\|$. 
		
		Applying \Cref{lem:old} with $i=3$, the total weight of the edges $e$ of $P_{ab}$ with $\dir(ab,e)\leq 3\cdot \sqrt{\eps}$ is at least $\frac79\, \|ab\|$. The parts of  these edges lying outside of $R(a,b)$ have weight at most $\|P_{ab}\setminus R(a,b)\|
		\leq \frac{1+3\eps}{2}\, \|ab\|$. Consequently, the remaining part of these edges are in $R(a,b)$, and their weight is at least $\left(\frac79 - \frac{1+3\eps}{2}\right)\|ab\| \leq \frac{7-27\eps}{18} \|ab\|\leq \frac29\, \|ab\|$ if $\eps<\frac19$, as claimed 
	\end{proof}
	
	We also need an observation from elementary geometry; see \Cref{fig:ellipse}. 
	
	\begin{lemma}\label{lem:elementary}
		For $a,b\in \mathbb{R}^d$, let $cd$ be the minor axis of the ellipsoid $\mathcal{E}_{ab}$. Then $\angle cad\leq 2\,\eps^{1/2}$.
	\end{lemma}
	
	\begin{proof}
		We may assume w.l.o.g.\ that $\|ab\|=1$.
		Let $o$ be the center of the ellipsoid $\mathcal{E}_{ab}$. 
		Then $\sec \angle cao =(\cos\angle cao)^{-1}=\frac{\|ac\|}{\|ao\|}= 1+\eps$.
		From the Taylor estimate $\sec(x)= 1+\frac12\,x^2+\frac{5}{24}x^4+\ldots \leq 1+x^2$ for $0<x<1$, we have 
		$\angle cao\geq \eps^{1/2}$. Consequently, 
		$\angle cad =2\, \angle cao \geq 2\eps^{1/2}$.
	\end{proof}
	
	\begin{theorem}
		\label{thm:UB2}
		Let $d=2$ and $\eps\in (0,1)$. The online algorithm $\alg_1$ maintains, for a sequence of $n$ points in Euclidean plane, an $(1+\eps)$-spanner of weight $O(\eps^{-3/2}\log\eps^{-1}\log n)\cdot \opt$, where $\opt$ denotes the minimum weight of an $(1+\eps)$-spanner for the same point set.
	\end{theorem}
	\begin{proof}
		\Cref{thm:UB} has established that algorithm~$\alg_1$ maintains a $(1+\eps)$-spanner. The tighter competitive analysis uses \Cref{lem:ellipse,lem:elementary}.
		
		\paragraph{Competitive Analysis.}
		Assume w.l.o.g.\ that $\diam(S_n)=\Theta(1)$, hence the side length of every quadtree square at level $\ell$ is $\Theta(2^{-\ell})$. 
		For a set $S_n=\{s_1,\ldots , s_n\}\subset \R^2$, let $G^*=(S_n,E^*)$ be a $(1+\eps)$-spanner of minimum weight, and let $\opt=\|G^*\|$. 
		Let $G=(S_n,E)$ be the spanner returned by the online algorithm $\alg_1$. Recall that $G=\bigcup_{\ell\geq 0}G_\ell$,
		where the total weight of all edges at levels $\ell> 2\log n$ is less than $\diam(S_n)$, so it is enough to consider $\ell=0,\ldots ,\lceil 2\log n\rceil$.
		
		\begin{claim}\label{cl:1}
			$\|G_\ell\|\leq O(\eps^{-3/2}\log\eps^{-1})\cdot \opt$ for all $\ell\geq 0$.
		\end{claim}
		\Cref{cl:1} immediately implies 
		$\|G\|\leq O(\eps^{-3/2}\log\eps^{-1}\log n)\cdot \opt$. For every level $\ell\geq 0$, $G_\ell=(P_\ell,E_\ell)$ is a graph on the representatives $P_\ell$. Note that $G^*$ is a Steiner spanner with respect to the point set $P_\ell$, as $G^*$ is a spanner on all $n$ points of the input. 
		
		We prove \Cref{cl:1} using a charging scheme: We charge the weight of every edge in $G_\ell$ to $G^*$ (more precisely, to line segments along the edges of $G^*$), and then show that each line segment of weight $w$ in $G^*$ receives $O(\eps^{-3/2}\log\eps^{-1})\cdot w$ charge. 
		
		For every point $p\in P_\ell$, algorithm $\alg_1$ greedily covers $\R^2$ by $\Theta(\eps^{-1})$ cones of aperture $\pi/k=\Theta(\eps^{-1})$ 
		and apex $p$, and adds an edge $pq_i$ in each nonempty cone $C_i$.
		For the competitive analysis, we greedily cover $\R^2$ by $\Theta(\eps^{-1/2})$ cones of aperture $\sqrt{\eps}$ and apex $p$. 
		We use translates of the same cone cover for all $p\in P_\ell$.
		Standard volume argument implies that a cone of aperture $\sqrt{\eps}$ intersects $O(\eps^{-1/2})$ cones of aperture $\Theta(\eps^{-1})$. We describe the charging scheme for each such cone $\widehat{C}$. 
		
		\paragraph{Charging Scheme.}
		Consider a cone $\widehat{C}$ with apex $p$ and aperture $\sqrt{\eps}$. 
		Let $E(\widehat{C})$ be the set of edges $pq$, $q\in \widehat{C}$ that algorithm $\alg_1$ adds to $G_\ell$ when $p$ is inserted into $P_\ell$. 
		Since $\widehat{C}$ intersects $O(\eps^{-1/2})$ cones of the ordered Yao-graph, then $|E(\widehat{C})|\leq O(\eps^{-1/2})$. 
		By construction, every edge in $G_\ell$ has weight at most $O(\eps^{-1}2^{-\ell})$.
		\begin{equation}\label{eq:c1}
			\|E(\widehat{C})\|
			=\sum_{pq\in E(\widehat{C})} \|pq\|
			\leq  |E(\widehat{C})| \cdot O(\eps^{-1}2^{-\ell})
			\leq O(\eps^{-3/2}2^{-\ell}).
		\end{equation}
		
		Let $q_0=q_0(\widehat{C})$ be a closest point in $P_\ell \cap \widehat{C}$ to $p$. (Possibly, $q_0$ arrived after $p$.) We distinguish between two cases:
		
		\noindent\textbf{Case~1: $\|pq_0\|< 2\cdot 2^{-\ell}$.} 
		Since $q_0\in P_\ell$, /nd $P_\ell$ contains at most one point in each quadtree cell of side length $\Theta(2^{-\ell})$, this case occurs for at most $O(1)$
		times per apex $p$. On the one hand, the sum of weights
		over all $p\in P_{\ell}$ and all cones $\widehat{C}$ with $\|pq_0\|< 2\cdot 2^{-\ell}$ is bounded by $O(|P_\ell|\cdot \eps^{-3/2}2^{-\ell})$.
		On the other hand, 
		$\opt\geq \Omega(\|\MST(P_\ell)\|)\geq \Omega( |P_\ell|\cdot 2^{-\ell})$.
		Consequently, the total weight of all edges handled in Case~1 is $O(\eps^{-3/2})\,\opt$.
		
		\noindent\textbf{Case~2: $\|pq_0\|\geq 2\cdot 2^{-\ell}$.} 
		The optimal spanner $G^*$ contains a $pq_0$-path $P_0$ of weight at most $(1+\eps)\|pq_0\|$. Recall $P_0$ lies in the ellipse $\mathcal{E}_0$ with foci $p$ and $q_0$, and $R(p,q_0)$ is the half of $\mathcal{E}_0$ that contains $q_0$ (cf.~\Cref{fig:ellipse}). Let $E^*(\widehat{C})$ be the set of maximal line segments $e$ along edges in $E^*$ such that $e\subset P_0\cap R(p,q_0)$ and 
		$\angle(e,pq_0)\leq 3\cdot \sqrt{\eps}$. 
		By \Cref{lem:ellipse}, we have $\|E^*(\widehat{C})\|\geq \frac13\|pq_0\|$.
		We distribute the weight of all edges in $E(\widehat{C})$ uniformly among the line segments in $E^*(\widehat{C})$. That is, each segment of weight $w$ in $E^*(\widehat{C})$ receives a charge of
		\begin{equation}\label{eq:111}
			\frac{\|E(\widehat{C})\|}{\|E^*(\widehat{C})\|}\cdot w 
			\leq \frac{ O(\eps^{-3/2}2^{-\ell})}{\Omega(2^{-\ell})} \cdot w
			\leq O(\eps^{-3/2})\cdot w.
		\end{equation}
		This completes the description of the charging scheme in Case~2. 
		
		\begin{figure}[htbp]
			\centering
			\includegraphics[width=0.95\textwidth]{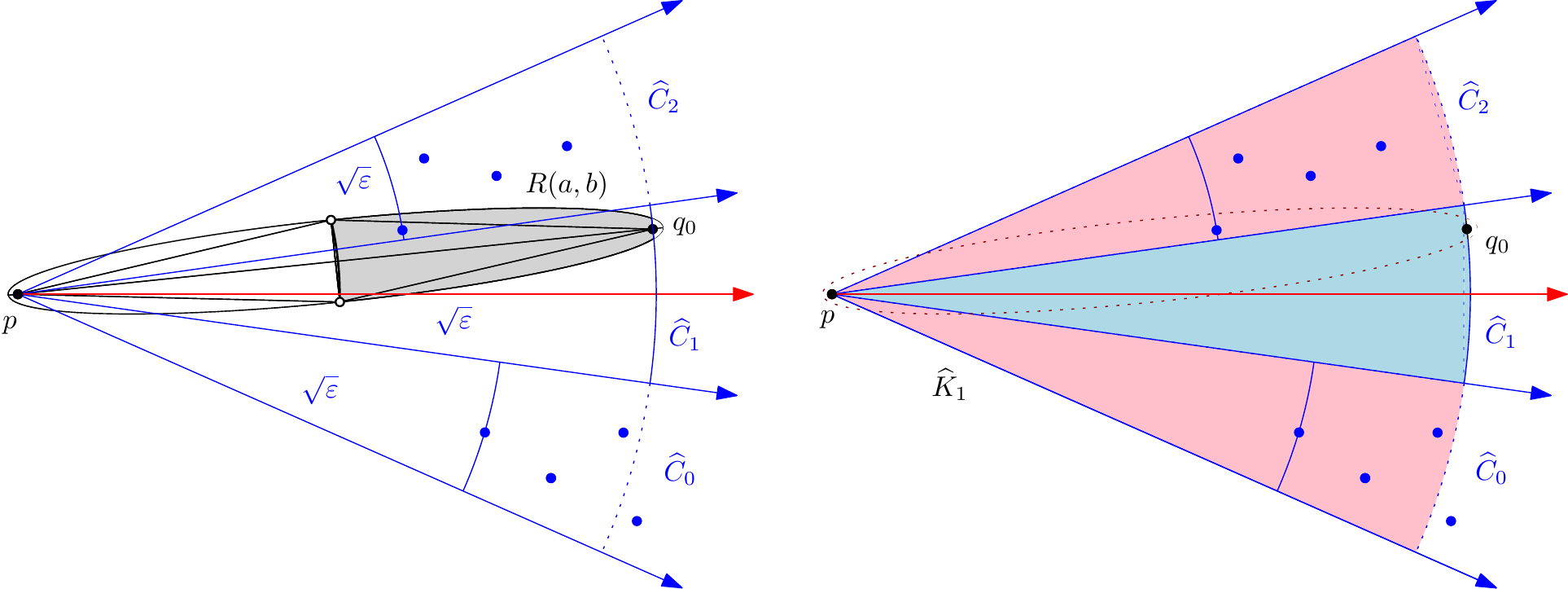}
			\caption{Left: There consecutive cones, $\widehat{C}_{0}$, $\widehat{C}_1$, and $\widehat{C}_{1}$, with apex $p$ and aperture $\sqrt{\eps}$. Point $q_0$ is the closest to $p$ in $P_\ell\cap \widehat{C}_1$; and  $R(p,q_0)\subset \widehat{K}_1=\widehat{C}_{0}\cup \widehat{C}_1\cup \widehat{C}_{2}$.
				Right: No point in $P_\ell$ is in the blue sector $\widehat{K}$, but there may be points in the pink sectors.}
			\label{fig:cones}
		\end{figure}
		
		\paragraph{Charges Received.}
		A point along an edge of the optimal spanner $G^*$ may receive charges from several cones $\widehat{C}$, possibly with different apices $p\in P_{\ell}$. 
		Let $L$ be a maximal line segment along an edge of $G^*$ such that every point in $L$ receives the same charges. 
		
		For a cone $\widehat{C}$ of aperture $\sqrt{\eps}$, let $\widehat{K}$ denote a cone with the same apex and axis as $\widehat{C}$, but aperture $3\,\sqrt{\eps}$; refer to \Cref{fig:cones}.
		\begin{claim}\label{cl:ellipse}
			If $L$ receives charges from $\widehat{C}$, then $L\subset \widehat{K}$.
		\end{claim}
		Indeed, if $L$ receive charges from $\widehat{C}$, then $L\subset R(p,q_0)\subset E_0$, where $\mathcal{E}_{0}$ is the ellipse with foci $p$ and the closest point $q_0\in \widehat{C}\cap P_\ell$. By \Cref{lem:elementary},
		$R(p,q_0)$ lies in a cone with apex $p$, aperture $2\sqrt{\eps}$, and axis $pq_0$.
		Consequently $L\subset R(p,q_0)\subset \widehat{K}$, which proves \Cref{cl:ellipse}.
		
		Note that if $L$ receives positive charge from a cone $\widehat{C}$ with apex $p$ and closest point $q_0$, then $\angle(L,pq_0)\leq 3\cdot \sqrt{\eps}$.
		Since the aperture of the cones $\widehat{C}$ is $\sqrt{\eps}$, then $L$ receives charges from cones $\widehat{C}$ with at most $O(1)$ different orientations. We may restrict ourselves to cones $\widehat{C}$ that are translates of each other (but have different apices in $P_{\ell}$).
		
		Let $\mathcal{A}$ be the set of all translates of a cone $\widehat{C}$ with aperture $\sqrt{\eps}$ and apices in $P_{\ell}$, and $L$ receives positive charge from $\widehat{C}$. We partition $\mathcal{A}$ into $O(\log \eps^{-1})$ classes as follows. For $j=1,\ldots ,\lceil \log(2\eps^{-1})\rceil$, let $\mathcal{A}_j$ be the set of cones $\widehat{C}\in \mathcal{A}$ such that $2^{j-\ell}\leq \|pq_0\|< 2^{j+1-\ell}$, where $p\in P_\ell$ is the apex of $\widehat{C}$ and $q_0$ is the closest point in $P_\ell\cap \widehat{C}$ to $p$.
		
		\begin{claim}\label{cl:2}
			For each $j$, segment $L$ receives $O(\eps^{-3/2})\, \|L\|$ total charges from all cones in $\mathcal{A}_j$.
		\end{claim}
		By refining~\eqref{eq:111} for a cone in $\widehat{C}\in \mathcal{A}_j$, we see that $L$ receives a charge
		\begin{equation}\label{eq:222}
			\frac{\|E(\widehat{C})\|}{\|E^*(\widehat{C})\|}\cdot \|L\|
			\leq \frac{ O(\eps^{-3/2}2^{j-\ell})}{\Omega(2^{-\ell})} \cdot \|L\|
			\leq O(\eps^{-3/2}2^{-j})\cdot \|L\|
		\end{equation}
		from each cone in $\mathcal{A}_j$.
		To prove \Cref{cl:2}, it is enough to show that $|\mathcal{A}_j|\leq O(2^j)$. 
		\begin{figure}[htbp]
			\centering
			\includegraphics[width=0.5\textwidth]{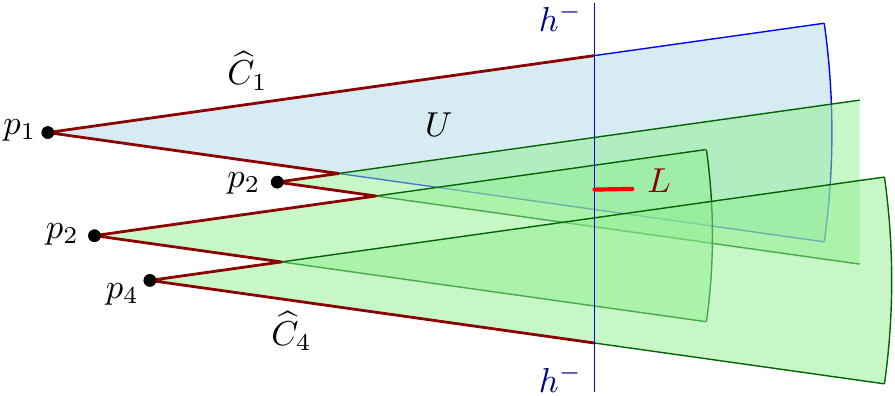}
			\caption{The union $U$ of triangles $\widehat{C}\cap h^-$, where $L$ receives charges from the cones $\widehat{C}$.}
			\label{fig:sectors}
		\end{figure}
		
		We may assume w.l.o.g. that the symmetry axis of every cone in $\mathcal{A}_j$ is parallel to the $x$-axis, and their apex is their leftmost point.
		Let $h$ be a vertical line that contains the left endpoint of $L$, and let $h^-$ be the left halfplane bounded by $h$; see \Cref{fig:sectors}. 
		The intersections $\widehat{C}\cap h$ and $\widehat{K}\cap h$ are vertical line segment of length $O(2^{j-\ell}\tan \sqrt{\eps})$. 
		We have $L\cap h\subset \widehat{K}\cap h$ by \Cref{cl:ellipse}; and obviously $\widehat{C}\cap h\subset \widehat{K}\cap h$. 
		Consequently, a vertical line segment of length $O(2^{j-\ell}\tan\sqrt{\eps})$ contains $h\cap \widehat{C}$ for all $\widehat{C}\in \mathcal{A}_j$.

		Let $U$ be the union of the triangles $\widehat{C}\cap h^-$ for all $\widehat{C}\in \mathcal{A}_j$. The interior of the $\widehat{C}\cap h^-$ does not contain any point in $P_\ell$. Consequently, the apices of all cones lie on the boundary $\partial U$ of $U$. The part of $\partial U$ in $h^-$ is a $y$-monotone curve with slopes $\pm \sqrt{\eps}$. It follows that the length of $\partial U$ is $O(2^{j-\ell}\tan \sqrt{\eps}/ \sin \sqrt{\eps})
		=O(2^{j-\ell} \csc\sqrt{\eps})
		=O(2^{j-\ell})$. This, in turn, implies that $\partial U$ intersects $O(2^{j})$ cubes of side length $a_0 2^{-\ell}$ at level $\ell$ of the quadtree, and so $|\mathcal{A}_j|\leq O(2^{j})$, as required. This completes the proof \Cref{cl:2}, and hence the proof of \Cref{thm:UB2}.
	\end{proof}

	\section{Lower Bounds in $\mathbb{R}^d$ Under the $L_1$ Norm}\label{sec:LB}
	
	In this section we introduce a strategy based on the points on the integer lattice $\mathbb{Z}^d$, that achieves a new lower bound for the competitive ratio of an online $(1+\epsilon)$-spanner algorithm in $\mathbb{R}^d$ under the $L_1$ norm.
	
	\old{Bhore and T\'oth~\cite{BT-oes-21} have recently established a lower bound of $\Omega(\eps^{-1}\log n / \log \eps^{-1})$ for the competitive ratio of the weight of an online spanner for points in one dimension. They also achieved a lower bound of $\Omega(\eps^{-2}/\log \eps^{-1})$ in $\mathbb{R}^2$ under the $L_1$-norm, and asked whether the exponent of $\epsilon$ in the lower bound increases with the dimension. 
		
		In this section we give a partial answer to Question~\Cref{Q2}. Specifically, we prove a lower bound of $\Omega_d(\eps^{-d})$ for the competitive ratio in $\mathbb{R}^d$ for constant $d$ under the $L_1$ norm.
		It remains an open problem whether the lower bound can be improved by a factor of $\log n/\log^{O(1)}\eps$, or whether the combination of the 
		current lower bounds $\Omega_d(\frac{\eps^{-1}}{\log\eps^{-1}}\,\log n+\eps^{-d})$ is tight for all constant dimensions $d\geq 1$.
	}
	
	\begin{figure}[htbp]
		\centering
		\includegraphics[width=0.4\textwidth]{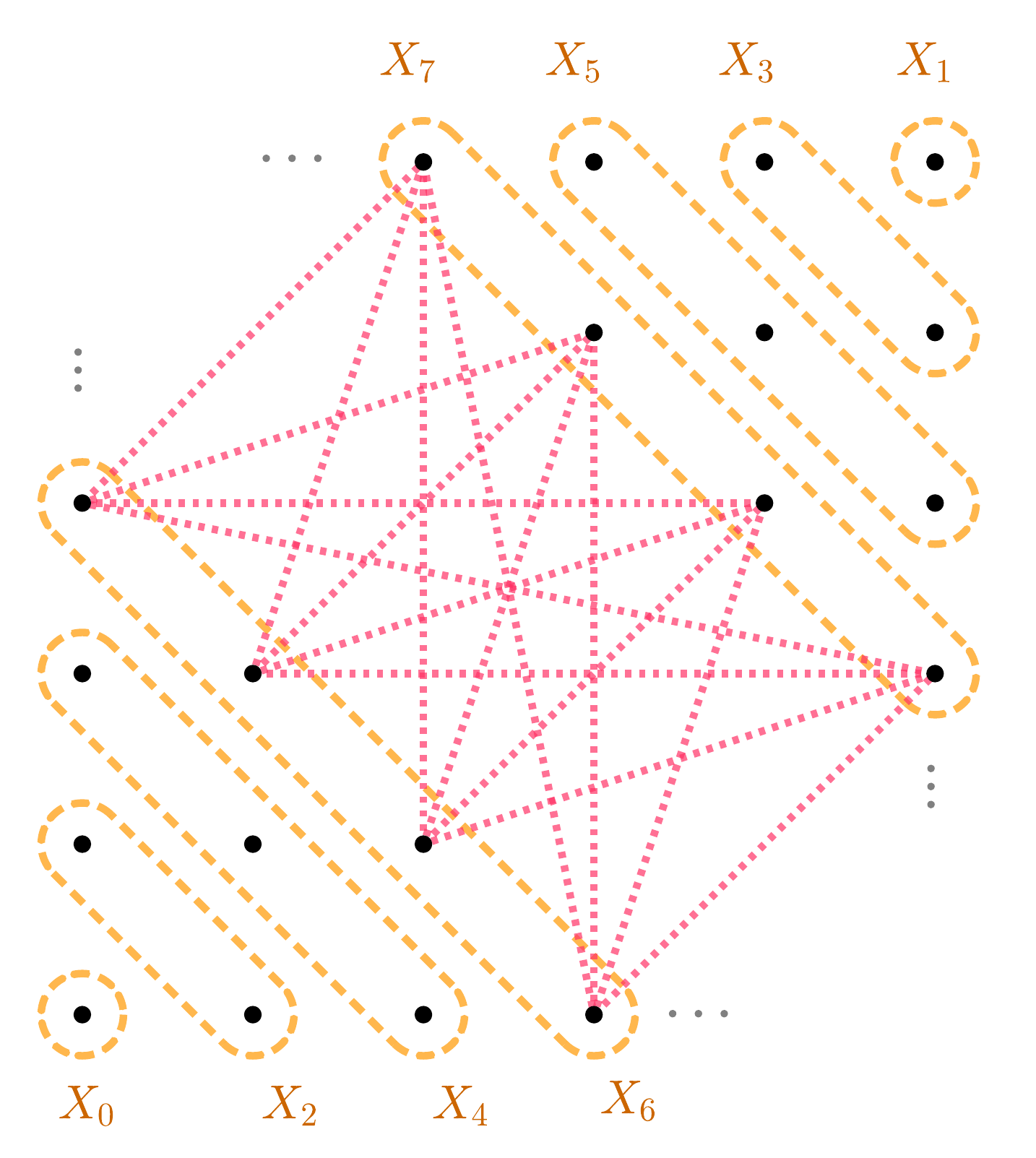}
		\caption{A sketch of the construction for the lower bound in two dimensions. Any online algorithm is required to add the red pairs.}
		\label{fig:lowerbd}
	\end{figure}
	
	\paragraph{Construction.}
	We describe an adversary strategy with $\Omega_d(\eps^{-d})$ points and show that any online algorithm returns a $(1+\eps)$-spanner whose weight is $\Omega_d(\eps^{-d})$ times the optimum weight. One can extend this result for arbitrary number of points, but that does not necessarily improve the lower bound. 
	The final point set $X$ consists of the points of the integer lattice $\mathbb{Z}^d$ in the hypercube $[0,\frac{1}{\eps d})^d$, where $\eps<\frac{1}{d}$. The points are presented in stages in order to deceive the online algorithm to add more edges than needed. In step $2i$, where $0\leq i< \frac{1}{2\eps}$, points $x\in X$ such that $\lVert x\rVert_1=i$ will be given to the algorithm. In step $2i+1$, where $0\leq i< \frac{1}{2\epsilon}$, the adversary presents points $x\in X$ such that $\lVert x\rVert_1=\lceil 1/\eps \rceil - i$ (\Cref{fig:lowerbd}). In other words, points are presented in batches according to their $L_1$ norms.
	
	\paragraph{Competitive Ratio.}
	Denote by $X_i$ the set of points presented in step $i$. The idea is to show that there has to exist many edges between $X_i$ and $X_{i+1}$ in order to guarantee the $1+\eps$ stretch-factor. Specifically, we define an \emph{ordered-pair} as follows.
	
	\begin{definition}[ordered-pair]
		A pair of points $(x,y)$ in $\mathbb{R}^d$ is an \emph{ordered-pair} if $x\in X_{2i}$ and $y\in X_{2i+1}$ for some $i$, and $x_k\leq y_k$ for all $k$, where $x_k$ and $y_k$ are the $k$-th coordinates of $x$ and $y$ respectively.
	\end{definition}
	Now we show that any ordered-pair $(x,y)\in X_{2i}\times X_{2i+1}$ requires an edge in the spanner immediately after $x$ and $y$ are presented. To prove this, we show that previously presented points cannot serve as via points in a $(1+\eps)$-path between $x$ and $y$.
	
	\begin{restatable}{lemma}{LBLemmaFirst}\label{lem:lowerbdop}
		Let $(x,y)$ be an ordered-pair. Then there is no $(1+\eps)$-path between $x$ and $y$ that goes through any other point $z\in X_j$ with $j\leq i+1$.
	\end{restatable}
	\begin{proof}
		Let $x_k$, $y_k$, and $z_k$ be the $k$-th coordinate of $x$, $y$, and $z$, respectively. Then the equality $\lVert x-z\rVert_1 + \lVert y-z\rVert_1 = \lVert x-y\rVert_1$ holds if and only if $x_k\leq z_k\leq y_k$, for all $k$. Since $z\neq x$ and $z\neq y$, we can conclude that $\lVert x\rVert_1 < \lVert z\rVert_1 < \lVert y\rVert_1$, which means that $z$ is not added in the previous steps, which is a contradiction. So the equality does not hold and $\lVert x-z\rVert_1 + \lVert y-z\rVert_1$ is strictly larger than $\lVert x-y\rVert_1$. As both expressions are integers, we have
		\begin{align*}
			\lVert x-z\rVert_1 + \lVert y-z\rVert_1 &\geq 1+\lVert x-y\rVert_1 \\
			&> \epsilon \lVert x-y\rVert_1 + \lVert x-y\rVert_1\\
			&=(1+\eps)\lVert x-y\rVert_1.
		\end{align*}
		The second inequality follows from the fact that $\lVert x-y\rVert_1<\eps^{-1}$
		which holds for any two points in $X$. The above inequality shows that a $(1+\eps)$-path between $x$ and $y$ cannot go through $z$ and completes the proof of the lemma.
	\end{proof}

	We next show that the total weight of the edges between ordered pairs is $\Omega_d(\epsilon^{-2d})$.
	
	\begin{restatable}{lemma}{LBLemmaSecond}\label{lem:lowerbdl1}
		The total weight of the edges between the ordered-pairs is $\Omega_d(\epsilon^{-2d})$. 
	\end{restatable}
	\begin{proof}
		Let $x=(x_1,\ldots,x_d)$ and $y=(y_1,\ldots,y_d)$ be two points in $X$. We show that if $x_k\in[\frac{1}{4\epsilon (d+0.25)}, \frac{1}{4\epsilon d}]$ for all $1\leq k\leq d$, and $y_k\in[\frac{3}{4\epsilon (d+0.25)}, \frac{3}{4\epsilon d}]$ for all $1\leq k\leq d-1$, then there is choice of $y_d$ that makes $(x,y)$ an ordered-pair. This would imply that there are $\Omega_d(\epsilon^{-2d+1})$ ordered-pairs and by \Cref{lem:lowerbdop}, each pair requires an edge of weight $\Omega_d(\eps^{-1})$, thus the total weight of required edges would be $\Omega_d(\epsilon^{-2d})$.
		
		In order to find such a $y_d$, recall that $\lVert x\rVert_1 + \lVert y\rVert_1 = \lceil \eps^{-1}\rceil$ holds because $(x,y)$ is an ordered-pair. This equality uniquely determines the value of $y_d$,
		\[
		y_d = \lceil \eps^{-1}\rceil - \sum_{k=1}^d x_k - \sum_{k=1}^{d-1} y_k.
		\]
		We just need to prove the inequalities $y_k\geq x_k$ and $y_k\leq 1/(\epsilon d)$ for this unique $y_k$. This can simply be done by plugging the maximum (and minimum) values of $x_k$s and other $y_k$s and calculating the result,
		\[
		y_d 
		\geq \frac{1}{\eps} - \frac{d}{4\epsilon d} - \frac{3(d-1)}{4\epsilon d}
		=\frac{3}{4\epsilon d}> x_d.
		\]
		Also,
		\[
		y_d 
		\leq \frac{1}{\eps}+1 - \frac{d}{4\epsilon (d+0.25)} - \frac{3(d-1)}{4\epsilon (d+0.25)}
		=1+\frac{1}{\epsilon (d+0.25)} < \frac{1}{\epsilon d}.
		\]
		
	\end{proof}
	
	Now we can prove the main theorem of this section.

	\begin{theorem}\label{thm:LB-L1}
		
		The competitive ratio of any online $(1+\eps)$-spanner algorithm in $\R^d$ under the $L_1$-norm is $\Omega_d(\epsilon^{-d})$. 
	\end{theorem}
	\begin{proof}
		For the point set $X\subset \R^d$, the unit-distance graph is a Manhattan network: It contains a path of weight $\|xy\|_1$ for all $x,y\in X$. Its weight is $\Theta_d(\eps^{-d})$ which is an upper bound for the weight of a $(1+\eps)$-spanner for any $\eps\geq 1$. 
		By \Cref{lem:lowerbdl1}, any online algorithm returns a spanner of weight $\Omega_d(\epsilon^{-2d})$. Thus its competitive ratio is $\Omega_d(\epsilon^{-d})$.
	\end{proof}

	\section{General Metrics: The Ordered Greedy Spanner}
	\label{sec:metric}
	In this section we study the online spanners problem on general metric spaces.
	The points arrive one by one, where for each new point we
	also receive its distances to all previously introduced points.

	In the offline setting, the celebrated greedy spanner algorithm~\cite{althofer1993sparse} sorts the edges by increasing weight, and then processes them one by one, adding each edge if by the time of examination, the distance between its endpoints is too large. This 
	algorithm achieves the existentially optimal\footnote{Specifically, if a $t$-spanner construction achieves an upper bound $m(n, t)$ and $l(n, t)$, resp., on the size and lightness of an $n$-vertex graph then this bound also holds for the greedy $t$-spanner~\cite{FS20}.}
	sparsity and lightness as a function of the stretch factor~\cite{FS20}. However, in the online model, we do not receive the edges in a sorted order, and therefore cannot execute the greedy algorithm. 
	As an alternative, we propose here the \emph{ordered greedy} algorithm. This is 
	a deterministic algorithm working against an adaptive adversary. The algorithm receives a stretch factor $t$, and works naturally as follows: We maintain a spanner $H$. When a point $v_{i}$ arrives, we order its edges\footnote{By edges we mean point pairs in the metric space, we will often use notation from graph theory.} in the original metric by weight. 
	Each edge $\{v_{i'},v_{i}\}$ is added to the spanner $H$ if currently $d_H(v_{i'},v_{i})>t\cdot d_X(v_{i'},v_{i})$.
	Note that this algorithm can be easily executed in an online fashion.
	
	\begin{theorem}
		\label{thm:GreedyOnlineSpanner}
		Given an $n$-point metric space $(X,d_X)$ in an (adaptive) adversarial order, with
		stretch factor $t=(2k-1)(1+\epsilon)$ for $k\geq 2$ 
		and $\epsilon\in(0,1)$, the ordered greedy algorithm returns a spanner with
		$O(\epsilon^{-1}\log\frac{1}{\epsilon})\cdot n^{1+\frac{1}{k}}$ edges
		and weight $O(\eps^{-1} n^{\frac{1}{k}}\log^{2}n)\cdot w(\MST)$.
	\end{theorem}
	
	\begin{proof}
		The bounded stretch of our spanner is straightforward by construction,
		as every pair was examined at some point, and taken care of. Next
		we analyze the lightness. 
		
		In the online spanning tree problem, points of a finite metric space 
		arrive one-by-one, and we need to connect each new point to a previous point to maintain a spanning tree. The ordered greedy algorithm connects each vertex $v_{i}$, to the closest vertex in $\{v_{1},\dots,v_{i-1}\}$. As was shown by Imase and Waxman~\cite{IW91}, the tree created by the ordered greedy algorithm has lightness $O(\log n)$, which is the best possible~\cite{IW91}.
		Denote the online spanning tree by $T_{G}$. Note that the ordered greedy spanner $H$ will contain $T_{G}$, as a shortest edge between a new vertex to a previously introduced vertex is always added to the spanner $H$.
		The following clustering lemma is frequently used for spanner constructions (see e.g.~\cite{ADFSW22, CW18, ElkinS16}). We provide a proof for the sake of completeness.
		
		\begin{claim}\label{cl:clusters}
			For every $i\in \N$, the point set $X$ can be partitioned into clusters $\mathcal{C}_{i}$ of diameter at most $D_{i}=\epsilon\cdot(1+\epsilon)^{i}$
			w.r.t.\ the metric $d_{T_{G}}$ such that $|\mathcal{C}_{i}|=O(\frac{w(T_{G})}{\epsilon\cdot(1+\epsilon)^{i}})$.
		\end{claim}
		\begin{proof}
			Let $N_{i}$ be a maximal set of vertices such that for every $x,y\in N_{i}$,
			$d_{T_{G}}(x,y)>\frac{1}{2}\cdot D_{i}$. For every vertex $x\in N_{i}$
			let $C_{x}=\left\{ z : x=\text{argmin}_{y\in N_{i}}d_{X}(z,y)\right\} $
			be the Voronoi cell of $x$. Clearly, $\diam(C_{x})\leq D_i$ for all $x$. Further, consider a continuous version of $T_{G}$ (where
			each edge is an interval). Then as the graph $T_{G}$ is connected,
			each cluster $C_{x}$ contains at least $\frac{1}{4}D_{i}$ length
			of edges (as the balls $\left\{ B_{T_{G}}(x,\frac{1}{4}D_{i})\right\} _{x\in N_{i}}$ are pairwise disjoint). It follows that 
			\[
			\left|\mathcal{C}_{i}\right|=\left|N_{i}\right|\le\frac{w(T_{G})}{\frac{1}{4}D_{i}}=O\left(\frac{w(T_{G})}{\epsilon\cdot(1+\epsilon)^{i}}\right),
			\]
			as claimed.
		\end{proof}
		For every $i$, consider the \emph{scale} $E_{i}=\left\{ e=\{u,v\}\in H : (1+\epsilon)^{i-1}\le d_{X}(u,v)<(1+\epsilon)^{i}\right\} $.
		We now ready to bound the lightness.
		
		\begin{claim}\label{cl:lightness}
			The weight of the ordered greedy spanner is $O(n^{\frac{1}{k}}\cdot\epsilon^{-2}\log^{2}n)\cdot w(\MST)$.
		\end{claim}
		\begin{proof}
			For scale $i$, consider the clusters $\mathcal{C}_{i}$ from \Cref{cl:clusters}. We create an (unweighted) \emph{cluster graph} $\mathcal{G}_{i}$ by contacting all the edges in each cluster and adding the edges $E_{i}$ (i.e., for every $\{u,v\}\in E_{i}$ such that $u\in C_{u}$ and $v\in C_{v}$, we add the edge $\left\{ c_{u},c_{v}\right\}$ to $\mathcal{G}_{i}$.
			Consider a cluster $C\in\mathcal{C}_{i}$ where $C=(u_{1},u_{2},\ldots,u_{|C|})$ are the vertices ordered w.r.t.\ arrival times. We argue that for every $j=1,\ldots , |C|$, the induced subgraph $T_{G}[\{u_{1},\dots,u_{j}\}]$ is connected. Assume for contradiction otherwise, and let $j$ be the first index
			violating this rule. Let $T_{G}^{j}$ be the tree $T_{G}$ right after
			the arrival of $u_{j}$. On the one hand, $T_{G}^{j}$ is connected, and so it contains a path $P$ from $u_{j}$ to $\{u_{1},\ldots,u_{j-1}\}$. By the
			assumption that $T_{G}[\{u_{1},\ldots,u_{j}\}]$ is disconnected, the a path $P$ has interior vertices that are not $\{u_{1},\ldots,u_{j}\}$.
			On the other hand, there is a path $P'$ from $u_{j}$ to $\{u_{1},\dots,u_{j-1}\}$ in $T_{G}[C]$. We conclude that $T_{G}$ contains two different paths from $u_{1}$ to $u_{j}$, a contradiction to the fact that $T_{G}$ is a tree. Furthermore,
			note that as $T_{G}$ is a tree, the diameter of $T_{G}[\{u_{1},\dots,u_{j}\}]$
			is bounded as well by $D_{i}$.
			
			We next argue that $\mathcal{G}_{i}$ is a simple graph. Suppose for
			contradiction that there is a cluster $C\in\mathcal{C}_{i}$ with
			a self loop. This implies that there are $v_{a},v_{b}\in C$ such
			that $\{v_{a},v_{b}\}\in E_{i}$. But this is impossible as $d_{X}(v_{a},v_{b})\le d_{T_{_{G}}}(v_{a},v_{b})<D_{i}=\epsilon\cdot(1+\epsilon)^{i}$.
			Next, suppose for contradiction that there is an edge $\{C,C'\}$ in $\mathcal{G}_{i}$ of multiplicity two or higher. Then there are vertices $x_{1},x_{2}\in C$ and $y_{1},y_{2}\in C'$ such that $\{x_{1},y_{1}\},\{x_{2},y_{2}\}\in E_{i}$.
			Assume w.l.o.g.\ that $y_{2}$ is the last arriving vertex among $\{x_{1},x_{2},y_{1},y_{2}\}$. At the time $\{x_{2},y_{2}\}$ is examined by the ordered greedy algorithm, there are paths from $x_{1}$
			to $x_{2}$ and from $y_{1}$ to $y_{2}$ of weight at most $D_{i}$.
			As $\{x_{1},y_{1}\}$ were already added to $H$, the spanner contains
			a $x_{2}y_{2}$-path of weight at most $2D_{i}+d_{X}(x_{1},y_{1})\le2\cdot\epsilon\cdot(1+\epsilon)^{i}+(1+\epsilon)^{i}<t\cdot(1+\epsilon)^{i}\le t\cdot d_{X}(x_{2},y_{2})$, which 
			contradicts to the fact that the algorithm chose to add $\{x_{2},y_{2}\}$.
			We conclude that $\mathcal{G}_{i}$ is indeed a simple graph.
			
			Next, we argue that $\mathcal{G}_{i}$ has girth at least $2k+1$.
			Suppose for contradiction that there is a cycle $C_{0}C_{1}C_{2}\dots C_{\beta}C_{0}$ in $\mathcal{G}_{i}$ with $\beta\le2k-1$, where the edge $C_{j}C_{j+1}$ corresponds to the edge $\{x_{j},y_{j+1}\}\in E_{i}$, modulo $\beta$. Assume w.l.o.g.\ that the edge $\{x_{\beta},y_{0}\}$ was added last. Note that at the time the algorithm examines $\{x_{\beta},y_{0}\}$,
			for every $j$, there is a path in $H$ from $y_{j}$ to $x_{j}$
			of weight at most $D_{i}$. Denote by $\widehat{H}$ the spanner $H$
			at this time. We conclude that
			\begin{align*}
				d_{\widehat{H}}(y_{0},x_{\beta}) & \le\sum_{j=0}^{\beta}d_{\widehat{H}}(y_{j},x_{j})+\sum_{j=0}^{\beta-1}d_{\widehat{H}}(x_{j},y_{j})\\
				& \le(\beta+1)\cdot D_{i}+\beta\cdot(1+\epsilon)^{i}\\
				& \le(2k-1)(1+3\epsilon)\cdot(1+\epsilon)^{i-1}\le(2k-1)(1+3\epsilon)\cdot d_{X}(y_{0},x_{2k-1}),
			\end{align*}
			which contradicts the fact that the edge $\{x_{\beta},y_{0}\}$ was
			added to the algorithm. 
			
			A graph with girth $2k+1$ contains at most $O(n^{1+\frac{1}{k}})$
			edges (see e.g. \cite{Bol78}). Hence the total weight
			of all the edges in $E_{i}$ is bounded by 
			\[
			(1+\epsilon)^{i}\cdot|E_{i}|
			=O(\left|\mathcal{C}_{i}\right|^{1+\frac{1}{k}})\cdot(1+\epsilon)^{i}
			=O(n^{\frac{1}{k}})\cdot\frac{w(T_{G})}{\epsilon\cdot(1+\epsilon)^{i}}\cdot(1+\epsilon)^{i}
			=O(\eps^{-1}\, n^{\frac{1}{k}})\cdot w(T_{G}).
			\]
			Let $e_{\max}$ be the heaviest edge in $H$, and let $i_{\max}$ be the index 
			such that $\{x,y\}\in E_{i_{\max}}$. Note that for every scale $i\le i_{\max}-\alpha$ have weight at most 
			\[
			w(E_{i})\le{n \choose 2}\cdot(1+\epsilon)^{i}\le n^{2}\cdot w(e_{\max})\cdot(1+\epsilon)^{-\alpha}\le n^{2}\cdot w(T_{G})\cdot(1+\epsilon)^{-\alpha}\,.
			\]
			We conclude that the weight of the spanner is bounded by
			\begin{align*}
				w(H) & =\sum_{i\le i_{\max}}w(E_{i})=\sum_{i=i_{\max}-\log_{1+\epsilon}n^{2}}^{i_{\max}}w(E_{i})+\sum_{i<i_{\max}-\log_{1+\epsilon}n^{2}}w(E_{i})\\
				& \le\log_{1+\epsilon}n^{2}\cdot O\left(\eps^{-1} n^{\frac{1}{k}}\right)\cdot w(T_{G})+\sum_{j\geq 1}w(T_{G})\cdot(1+\epsilon)^{-j}\\
				& \le O\left(\frac{\log n}{\log (1+\eps)}\cdot \frac{n^{\frac{1}{k}}}{\eps} +\frac{1}{\epsilon}\right)\cdot w(T_{G})\\
				& =O\left(n^{\frac{1}{k}}\cdot\frac{\log n}{\epsilon^{2}}\right)\cdot w(T_{G})=O\left(n^{\frac{1}{k}}\cdot\frac{\log^{2}n}{\epsilon^{2}}\right)\cdot w(\MST)\,.
			\end{align*}
		\end{proof}
		We next bound the sparsity of the ordered greedy spanner.
		
		\begin{claim}
			\label{cl:sparsity}
			The ordered greedy spanner has $O(\epsilon^{-1}\log\frac{1}{\epsilon})\cdot n^{1+\frac{1}{k}}$ edges.
		\end{claim}
		\begin{proof}
			We will assume for simplicity that the algorithm was executed with parameter $t=(2k-1)(1+2\eps)$, later one can scale the results accordingly.
			Let $\{v_{1},\dots,v_{n}\}$ be the order in which the vertices arrived. Let $H_i$ be the state of the spanner just after the arrival of $v_i$.
			We will greedily construct a laminar set system  $N_{0}\subseteq N_{1}\subseteq\dots$,
			where every pair of point in $N_{i}$ will be at distance at least $(1+\epsilon)^{i}$ w.r.t. the spanner $H$.
			Specifically, given a newly arrived vertex $v_j$ which already joined $N_{i}$, $v_j$ will join $N_{i+1}$ if there is no vertex $v_{j'}$ (where $j'<j$) at distance $d_{H_j}(v_j,v_{j'})\le(1+\epsilon)^{i+1}$ in the current spanner. Let $\Delta_{i}=\frac{(1+\epsilon)^{i+1}-1}{\epsilon}$.
			We will call each set $N_i$ a \emph{net}, and every point $v_j\in N_i$ a net point.
			We argue that the set $N_{i}$ is $\Delta_{i}$ dominating, that is 
			every vertex $v_j$ has a net point $v_{j'}\in N_{i}$, such that at the time $v_j$ arrived, $d_{H_j}(v_j,v_{j'})\le\Delta_{i}$.
			
			Indeed, by induction there is a net point $v_q\in N_{i-1}$ such that
			$d_{H_j}(v_j,v_q)\le\Delta_{i-1}$, and $q<j$. If $v_q\in N_{i}$ then we are done. Otherwise,
			there is a point $v_p\in N_{i}$ such that $d_{H_q}(v_q,v_p)\le(1+\epsilon)^{i}$ and $s<q$. Implying $d_{H_{j}}(v_{j},v_{p})\le d_{H_{j}}(v_{j},v_{q})+d_{H_{q}}(v_{q},v_{p})\le\frac{(1+\epsilon)^{i}-1}{\epsilon}+(1+\epsilon)^{i}=\frac{(1+\epsilon)^{i+1}-1}{\epsilon}=\Delta_{i}$.
			For $i$ too small, let $N_i=X$, and $\Delta_i=0$.
			
			For every $i$, consider the \emph{scale} $E_{i}=\left\{ e=\{u,v\}\in H : (1+\eps)^{i-1}\le d_{X}(u,v)<(1+\eps)^{i}\right\}$.
			Set  $s=\left\lceil \log_{1+\epsilon}(\frac{4k}{2k-1}\cdot\frac{1+\epsilon}{\epsilon^{2}\cdot})\right\rceil$.
			
			We argue that for every $i$, 	$|E_i|\le O\left(|N_{i-s}\setminus N_{i+s}|\right)^{1+\frac1k}$.
			For this goal, we construct an auxiliary  graph $\mathcal{G}_i$ 
			with $N_{i-s}$ as vertices and $E_{i}$ as edges. Specifically, for every $\{x,y\}\in E_{i}$, let $v_x v_y\in N_i$ be the closest vertices to $x,y$ in $N_i$ at the time they were added. Then we will add the edge $\{v_x,v_y\}$ to $\mathcal{G}_i$.

			Clearly $\mathcal{G}_{i}$ does not contain self loops, as the distance between two vertices $x,y$ who has the same closest vertex in $N_i$ is bounded by $2\Delta_{i-s}<(1+\eps)^{i-1}$.
			Suppose for contradiction that there is an edge $\{v,u\}$ in $\mathcal{G}_{i}$ of multiplicity two or higher. Then there are vertices $x_{1},x_{2},y_{1},y_{2}$ such that $v$ was the closest vertex to $x_{1},x_{2}$, $u$ was the closest vertex to $y_{1},y_{2}$, and $\{x_{1},y_{1}\},\{x_{2},y_{2}\}\in E_{i}$.
			Assume w.l.o.g.\ that $y_{2}$ is the last arriving vertex among $\{x_{1},x_{2},y_{1},y_{2}\}$. At the time $\{x_{2},y_{2}\}$ is examined by the ordered greedy algorithm, the pairs $(x_1,x_2)$ and $(y_1,y_2)$ already were examined, and hence $H$ contain path from $x_{1}$ to $x_{2}$ and from $y_{1}$ to $y_{2}$ of weight at most  $2\cdot\Delta_{i-s}$.
			By our assumption, $\{x_{1},y_{1}\}$ was already added to $H$. Hence the spanner contains a $x_{2}y_{2}$-path of weight at most
			\begin{align*}
				d_{H}(x_{2},y_{2}) & \le d_{H}(x_{2},y_{1})+d_{X}(x_{1},y_{1})+d_{H}(y_{1},y_{2})\\
				& \le4\Delta_{i-s}+(1+\epsilon)^{i}\\
				& \le\frac{4(1+\epsilon)^{i-s}}{\epsilon}+(1+\epsilon)^{i}\\
				& \le\left(1+\epsilon+\frac{4}{\epsilon(1+\epsilon)^{s-1}}\right)(1+\epsilon)^{i-1}\le t\cdot d_{X}(x_{2},y_{2})~,
			\end{align*}
			a contradiction to the fact that the algorithm choose to add $\{x_2,y_2\}$.
			
			Next, we argue that $\mathcal{G}_{i}$ has girth at least $2k+1$.
			Suppose for contradiction that there is a cycle $u_{0}u_{1}u_{2}\dots u_{\beta}u_{0}$ in $\mathcal{G}_{i}$ with $\beta\le2k-1$, where the edge $u_{j}u_{j+1}$ corresponds to the edge $\{x_{j},y_{j+1}\}\in E_{i}$, modulo $\beta$. Assume w.l.o.g. that the edge $\{x_{\beta},y_{0}\}$ was added last. Note that at the time the algorithm examines $\{x_{\beta},y_{0}\}$,
			for every $j$, there is a path in $H$ from $y_{j}$ to $x_{j}$ of weight at most $2\cdot\Delta_{i-s}$. Denote by $\widehat{H}$ the spanner $H$ at this time. We conclude that
			\begin{align*}
				d_{\widehat{H}}(y_{0},x_{\beta}) & \le\sum_{j=0}^{\beta}d_{\widehat{H}}(y_{j},x_{j})+\sum_{j=0}^{\beta-1}d_{\widehat{H}}(x_{j},y_{j})\\
				& \le(\beta+1)\cdot2\Delta_{i-s}+\beta\cdot(1+\epsilon)^{i}\\
				& \le2k\cdot\frac{2(1+\epsilon)^{i-s}}{\epsilon}+(2k-1)\cdot(1+\epsilon)^{i}\\
				& =(2k-1)(1+\epsilon+\frac{4k}{2k-1}\cdot\frac{1}{\epsilon\cdot(1+\epsilon)^{s-1}})\cdot(1+\epsilon)^{i-1}\\
				& \le(2k-1)(1+2\epsilon)\cdot d_{X}(y_{0},x_{2k-1})~,
			\end{align*}
			which contradicts the fact that the edge $\{x_{\beta},y_{0}\}$ was
			added to the algorithm. 
			
			Consider a pair of net points $u,v\in N_{i+s}$. Then the distance between $u,v$ in $\mathcal{G}_i$ has to be at least $3$. Otherwise, if $d_{\mathcal{G}_i}(u,v)\le 2$,  there is a net point $z\in N_{i-s}$ and two edges $\{x_{0},y_{1}\},\{x_{1},y_{2}\}\in E_{i}$ corresponding to $\{u,z\},\{z,v\}$ in $\mathcal{G}_i$. Then following the logic above,
			\begin{align*}
				d_{\widehat{H}}(u,v) & \le d_{\widehat{H}}(u,x_{0})+d_{\widehat{H}}(x_{0},y_{1})+d_{\widehat{H}}(y_{1},x_{1})+d_{\widehat{H}}(x_{1},y_{2})+d_{\widehat{H}}(y_{2},v)\\
				& \le4\Delta_{i-s}+2\cdot(1+\epsilon)^{i}\\
				& \le\left(\frac{8}{\epsilon(1+\epsilon)^{s}}+2\right)\cdot(1+\epsilon)^{i}<(1+\epsilon)^{i+s}~,
			\end{align*} 
			a contradiction to the fact that both $u,v$ joined $N_i$. It follows that there are no edges between vertices in $N_{i+s}$, and furthermore, each vertex in $N_{i-s}\setminus N_{i+s}$ is connected to at most a single vertex in $N_{i+s}$. We conclude that the number of edges incident on $N_{i+s}$ vertices is bounded by $|N_{i-s}\setminus N_{i+s}|$. 
			As the induced graph $\mathcal{G}_i[N_{i-s}\setminus N_{i+s}]$ has girth $2k+1$, it contains at most $O\left(\left|N_{i-s}\setminus N_{i+s}\right|^{1+\frac1k}\right)$ edges (see e.g. \cite{Bol78}). 
			We conclude
			\[
			|E_{i}|=E\left(\mathcal{G}_{i}\right)=E\left(G\left[N_{i-s}\setminus N_{i+s}\right]\right)+\left|N_{i-s}\setminus N_{i+s}\right|=O\left(\left|N_{i-s}\setminus N_{i+s}\right|^{1+\frac{1}{k}}\right)~.
			\]
			
			We conclude a bound on the number of edges:
			\begin{align*}
				\left|E(H)\right| & =\sum_{i\ge0}|E_{i}|\le\sum_{i\ge0}O\left(\left|N_{i-s}\setminus N_{i+s}\right|^{1+\frac{1}{k}}\right)\\
				& \le O(n^{\frac{1}{k}})\cdot\sum_{i\ge0}\left|N_{i-s}\setminus N_{i+s}\right|=O(s\cdot n^{1+\frac{1}{k}})=O(\frac{\log\frac{1}{\epsilon}}{\epsilon}\cdot n^{1+\frac{1}{k}})~,
			\end{align*}
			where the second to last equality follows as each vertex can participate in at least $2s$ different addends  in the sum.
		\end{proof}
		The theorem now follows.
	\end{proof}

	\section{Lower Bound for General metrics}
	In this section we prove an $\Omega(\frac1k \cdot n^{\frac1k})$ lower bound on the competitive ratio of an online $(2k-1)$-spanner of $n$-vertex graphs. Our lower bound holds in both cases where the quality is measured by number of edges or the weight. It follows that our upper bound in \Cref{thm:GreedyOnlineSpanner} cannot be substantially improved, even if we consider competitive ratio instead of lightness/sparsity.
	
	Recall that the Erd\H{o}s Girth Conjecture~\cite{Erdos64} states that for every $n,k\ge 1$, there exists an $n$-vertex graph with $\Omega(n^{1+\frac1k})$ edges and girth $2k+2$.
	The proof of the following lemma is based on a counting argument form the recent lower bound proof for (static) vertex fault tolerant emulators by Bodwin, Dinitz, and Nazari \cite{BDN22}.
	
	\begin{lemma}\label{lem:smallDiamGirth}
		Assuming the Erd\H{o}s girth conjecture, for every $n,k\ge1$, there exists an
		$n$-point metric space $(X,d_{X})$ with diameter $2k-1$, such that
		every $(2k-1)$-spanner has $\Omega(\frac{1}{k}\cdot n^{1+\frac{1}{k}})$
		edges and weight $\Omega(n^{1+\frac{1}{k}})$.
	\end{lemma}
	\begin{proof}
		Let $G=(V,E_{G})$ be the graph fulfilling the Erd\H{o}s girth conjecture. That is, $G$ is an unweighted $n$-vertex graph with girth $2k+2$ and $|E_G|=\Omega(n^{1+\frac{1}{k}})$
		edges. Set a metric $d_X$ over $V$ as follows, \footnote{Note that $\forall x,y,z\in V$,
			$d_{X}(x,z)=\min\left\{ d_{G}(x,z),2k-1\right\} \le\min\left\{ d_{G}(x,y)+d_{G}(y,z),2k-1\right\} \le\min\left\{ d_{G}(x,y),2k-1\right\} +\min\left\{ d_{G}(y,z),2k-1\right\} =d_{X}(x,y)+d_{X}(y,z)$. Thus $d_{X}$ is a metric space.
		}
		\[
		\forall u,v\in V\,\quad d_{X}(u,v)=\min\left\{ d_{G}(u,v),2k-1\right\} .
		\]
		Suppose that $H=(V,E_{H})$ is a $(2k-1)$-spanner for $(V,d_{X})$ with weight function $w_H$, where the weight of an edge $e'\in\{u,v\}\in E_H$ is $w_H(e')=d_X(u,v)$.
		Let $E'=E_{H}\setminus E_{G}$ be the edges of $H$ which are not
		in $G$. 
		We say that an edge $e'\in E'$ \emph{covers} an edge
		$e\in E_{G}$, if there is a shortest path in $G$ between the endpoints of
		$e'$ going through $e$ of weight at most $k$. Note that as $e'$
		has weight at most $k$, there is a unique shortest path in $G$ between its
		endpoints. In particular, each edge $e\in E'$ can cover at most $k$
		edges in $E_{G}$.
		
		Consider an edge $e=\{v_0,v_s\}\in E_{G}\setminus E_{H}$.
		We argue that some edge $e'\in E'$ must cover $e$. Suppose for contradiction
		otherwise, and let $P=(v_{0},v_{1},\dots,v_{s})$ be the shortest path in $H$ between the
		endpoints $v_0,v_s$ of $e$. 
		Suppose first that $P$ contains an edge $v_{i},v_{i+1}$ of weight
		at least $w_H(\{v_i,v_{i+1}\})\ge k+1$. In particular, $d_G(\{v_i,v_{i+1}\})\ge k+1$.
		Then by the triangle inequality, $d_{G}(v_{0},v_{i})+d_{G}(v_{i+1},v_{s})\ge d_{G}(v_{i},v_{i+1})-d_{G}(v_{0},v_{s})\ge k$. It follows that $P$ has weight at least $2k+1$, a contradiction to the fact that $H$ is a $2k-1$ spanner.
		We conclude that for every $i\in\{0,\dots,s-1\}$, $d_X(v_i,v_{i+1})=d_G(v_i,v_{i+1})\le k$. In particular, in $G$ there is a unique path $P_i=(u_0^i,\dots,u_{s_i}^i)$ between $v_i$ to $v_{i+1}$ of weight $d_G(v_i,v_{i+1})\le k$. As no edge covers $e$, $e$ does not belong to any of these paths.
		The concatenation of this paths $P_0\circ P_1\circ\dots\circ P_{s-1}$ is a path in $G$ of at most $2k-1$ edges between the endpoints of $e$. It follows that $G$ contains a $2k$-cycle, a contradiction.
		
		For conclusion, as every edge in $E_{G}\setminus E_{H}$ is covered, and every edge
		in $E'=E_{H}\setminus E_{G}$ can cover at most $k$ edges, it follows
		that $|E_{H}\setminus E_{G}|\ge\frac{1}{k}\cdot\left|E_{G}\setminus E_{H}\right|$.
		In particular,
		\[
		|E_{H}|=|E_{H}\cap E_{G}|+|E_{H}\setminus E_{G}|\ge|E_{H}\cap E_{G}|+\frac{1}{k}\cdot\left|E_{G}\setminus E_{H}\right|\ge\frac{1}{k}\cdot\left|E_{G}\right| .
		\]
		To bound the weight, for each edge $e'=\{s,t\}\in E'$, let $A_{e'}$
		be the set of edges in $E_{G}$ covered by $e'$. Note that $w_{H}(e')=d_{G}(s,t)=|A_{e'}|$.
		As all the edges in $E_{G}\setminus E_{H}$ are covered, we conclude
		\begin{align*}
			w_{H}(E_{H}) & =w_{H}(E_{H}\cap E_{G})+w_{H}(E_{H}\setminus E_{G})\\
			& =\left|E_{H}\cap E_{G}\right|+\sum_{e'\in E'}|A_{e'}|\\
			& \ge\left|E_{H}\cap E_{G}\right|+\left|E_{G}\setminus E_{H}\right|=|E_{G}|=\Omega(n^{1+\frac{1}{k}}) ,
		\end{align*}
		the lemma now follows.
	\end{proof}

	\begin{theorem}\label{thm:GeneralMetricLB}
		Assuming Erd\H{o}s girth conjecture, the competitive ratio of any online $(2k-1)$-spanner algorithm for $n$-point  metrics is $\Omega(\frac{1}{k}\cdot n^{\frac{1}{k}})$, for both weight and edges.\\
		In more details, there is an $n$-point metric space $(X,d_X)$ with a $(2k-1)$-spanner $H_\opt=(X,E_\opt)$, and order over $X$ for which every $(2k-1)$-spanner produced by an online algorithm will have $\Omega(\frac{1}{k}\cdot n^{\frac{1}{k}})\cdot|E_\opt|$ edges, and  $\Omega(\frac{1}{k}\cdot n^{\frac{1}{k}})\cdot w(H_\opt)$ weight.
	\end{theorem}
	\begin{proof}
		Consider the metric space $(X,d_X)$ from \Cref{lem:smallDiamGirth} with parameters $n-1$ and $k$. Let $X'$ be the metric space $X$ with an additional point $r$ at distance $\frac{k-1}{2}$ from all the points in $X$. Note that no pairwise distance is changed due to the introduction of $r$. The adversary provides the online algorithm the points in $X$ first (in some arbitrary order), and the point $r$ last.
		After the algorithm received all the points in $X'$, it has a $2k-1$-spanner $H_{n-1}$.  According to \Cref{lem:smallDiamGirth}, $H_{n-1}$ has $\Omega(\frac{1}{k}\cdot (n-1)^{1+\frac{1}{k}})=\Omega(\frac{1}{k}\cdot n^{1+\frac{1}{k}})$ edges, and $\Omega(n^{1+\frac{1}{k}})$ weight.
		
		Next the algorithm introduces $r$. Consider the spanner $S=(X',E_S)$ consisting of $n-1$ edges with $r$ as a center. Note that the maximum distance in $S$ is $2k-1$, and hence $S$ is a $2k-1$ spanner as required. Note that $S$ contains $n-1$ edges of weight $\frac{2k-1}{2}$ each, and thus have total weight of $O(nk)$.
		We conclude
		\[\setlength\arraycolsep{3pt}\def\arraystretch{1.5}
		\begin{array}{ccccccc}
			|E_{H_{n}}| &~\ge~&|E_{H_{n-1}}|&~=~&\Omega(\frac{1}{k}\cdot n^{1+\frac{1}{k}})&~=~&\Omega(\frac{1}{k}\cdot n^{\frac{1}{k}})\cdot|E_{S}|~.\\
			w(E_{H_{n}}) &~\ge~& w(E_{H_{n-1}})&~=~&\Omega(n^{1+\frac{1}{k}})&~=~&\Omega(\frac{1}{k}\cdot n^{\frac{1}{k}})\cdot w(S)~.
		\end{array}
		\]
	\end{proof}

	\section{Ultrametrics}
	\label{sec:ultra}

	An ultrametric $\left(X,d\right)$ is a metric space satisfying a strong form of the triangle inequality, that is, for all $x,y,z\in X$,
	$d(x,z)\le\max\left\{ d(x,y),d(y,z)\right\}$. A related notion is a $k$-hierarchical well-separated tree ($k$-HST).
	
	\begin{definition}[$\alpha$-HST]\label{def:HST}
		A metric $(X,d_X)$ is a $\alpha$-hierarchical well-separated tree ($\alpha$-HST) if there exists a bijection $\varphi$ from $X$ to leaves of a rooted tree $T$ in which:
		\begin{itemize}
			\item Each node $v\in T$ is associated with a label $\ell(v)$ 
			such that $\ell(v) = 0$ if $v$ is a leaf and $\ell(v)\geq \alpha\ell(u)$ if $v$ is an internal node and $u$ is any child of $v$.
			\item $d_X(x,y) = \ell(\lca(\varphi(x),\varphi(y)))$ where $\lca(u,v)$ is the least common ancestor of any two given nodes $u,v$ in $T$. 
		\end{itemize}
	\end{definition}
	
	It is well known that any ultrametric is a $1$-HST, and any $k$-HST is an ultrametric~\cite{BLMN05}.
	
	Suppose that we are given an HST in the online model. Construct a spanner $H$ using the following algorithm: for every arriving vertex $v$, let $u$ be the first vertex in the order of arrival among all the nearest neighbors of $v$. We
	add the edge $\{u,v\}$ to the spanner $H$. Note that $H$ is a spanning tree at all times (we will later argue that it is actually an MST).

	We show  that for general ultrametrics, the online algorithm can maintain a spanner of lightness arbitrarily close to 1 (with constant stretch).
	
	\begin{lemma}
		\label{lem:alphaHstStretch}
		If $U$ is an $\alpha$-HST, then the spanner $H$ has distortion
		$2\cdot\frac{\alpha}{\alpha-1}$.
	\end{lemma}
	\begin{proof}
		Think of the representation of the HST as a tree with labeled internal
		nodes. For every internal node $\chi$, we call the first descendent in the order
		of arrival the \emph{center} of $\chi$. Consider a vertex $v$ at the time
		of its arrival, let $\chi$ be an internal node which is an ancestor
		of $v$, and let $u$ be the center of $\chi$. 
		We argue that $d_{H}(v,u)\le t\cdot d_{U}(v,u)$
		for $t=\frac{\alpha}{\alpha-1}$. The proof is by induction. 
		The induction step is immediate if the edge $\{u,v\}$ was added to $H$. 
		Otherwise, let $\chi'$ be the highest internal node which is an ancestor of $v$ but has a center other than $u$. Let $x$ be the center of $\chi'$. At
		the time when $x$ arrives, it was the only descendent of $\chi'$.
		In particular, the closest neighbors of $x$ at this time is $u$
		(as otherwise, there must be an internal vertex $\chi''$ between
		$\chi'$ and $\chi$ with center other than $u$). As $u$ is
		the center of $\chi$, it is the first arriving descendent of $\chi$.
		In other words, $u$ is the first vertex in the order of arrival among
		all the nearest neighbors of $u$. We conclude that $\{x,u\}\in H$.
		As $U$ is an $\alpha$-HST $\ell(\chi')\le\frac{1}{\alpha}\ell(\chi)$.
		By the induction hypothesis, $d_{H}(v,x)\le t\cdot d_{U}(v,x)$. We
		conclude
		\begin{align*}
			d_{H}(v,u) & \le d_{H}(v,x)+d_{H}(x,u)\le t\cdot d_{U}(v,x)+d_{U}(x,u)\\
			& =t\cdot\ell(\chi')+\ell(\chi)\le\left(\frac{t}{\alpha}+1\right)\cdot\ell(\chi)=t\cdot\ell(\chi)=t\cdot d_{U}(v,u)\,.
		\end{align*}
		For two arbitrary vertices $u,v$, let $\chi=\mathrm{lca}(u,v)$, and let $x$ be the center of $\chi$. By the definition of HST, $d_{U}(v,x),d_{U}(x,u)\le d_{U}(v,u)$.
		Using the previous argument,
		\[
		d_{H}(v,u)\le d_{H}(v,x)+d_{H}(x,u)\le t\cdot\left(d_{U}(v,x)+d_{U}(x,u)\right)=2t\cdot d_{U}(v,u)\,.
		\]
	\end{proof}

	\begin{lemma}
		\label{lem:GreedyHSTisMST}
		The spanner $H$ is an MST of $U$.
	\end{lemma}
	\begin{proof}
		Assume for contradiction otherwise. Then $w\left(\text{MST}(U)\right)<w\left(H\right).$
		Let $T$ be an MST of $U$ containing the maximum number edges of $H$. 
		Let $\{u,v\}=e\in H\setminus T$ be some edge. Assume w.l.o.g.\ that
		$u$ arrived before $v$, and let $\chi=\mathrm{lca}(u,v)$.
		As the algorithm added edge $\{u,v\}$ to $H$, necessarily $u$ is the center
		of $\chi$. Further, there is a child node $\chi_{v}$ of $\chi$,
		where $v$ is a unique descendent of $\chi_{v}$ (at the time of arrival).
		Let $S_{v}$ be the set of all descendants of $\chi_{v}$ in $U$. 
		Then $T$ contains at least one edge from the vertices of $S_{v}$ to a vertex outside of $S_{v}$. Let $e'\in T$ be such an edge that is on the unique
		$uv$-path in $T$. Then $w(e')\ge\ell(\chi)=w(e)$, and $T\cup\{e\}\setminus\{e'\}$
		is a spanning tree of $U$, of weight at most $w(T)$. A contradiction
		to the maximality of $T$. 
	\end{proof}

	\begin{theorem}
		\label{thm:ultrametricStretchAlpha}
		Given an ultrametric $U$, for every $\alpha\ge1$, 
		an online algorithm can maintain a 
		$\frac{2\alpha^{2}}{\alpha-1}$-spanner of wieght $\alpha\cdot w(\MST)$.
		Alternatively, for every $\eps>0$, it can maintain a spanner of weigth $(1+\eps)\cdot w(\MST)$ and stretch $\frac{2(1+\epsilon)^{2}}{\epsilon}=O(\eps^{-1})$.
	\end{theorem}
	\begin{proof}
		Let $U_{\alpha}$ be the $\alpha$-HST for $U$ where we round every distance up to
		the next integer power of $\alpha$. That is, $d_{U_{\alpha}}(u,v)=\alpha^{\left\lceil \log_{\alpha}d_{U}(u,v)\right\rceil }$.
		Note that $d_{U}(u,v)\le d_{U_{\alpha}}(u,v)<\alpha\cdot d_{U}(u,v)$.
		In particular, the weight of the MST in $U_{\alpha}$ is larger than
		the MST of $U$ by at most a factor $\alpha$. We run the online algorithm
		above on $U_{\alpha}$ instead of $U$. As a result, we get a spanner
		$H_{\alpha}$ of $U_{\alpha}$ with stretch $2\cdot\frac{\alpha}{\alpha-1}$
		and lightness $1$ (w.r.t.\ $U_{\alpha}$). Let $H$ be the same spanner
		with the original weights. Then for every pair of vertices $u,v$
		\[
		d_{H}(u,v)\le d_{H_{\alpha}}(u,v)\le\frac{2\alpha}{\alpha-1}\cdot d_{U_{\alpha}}(u,v)\le\frac{2\alpha^{2}}{\alpha-1}\cdot d_{U}(u,v)\,.
		\]
		The weight of $H$ is bounded by 
		$w(H)\le w(H_{\alpha})=w(\text{MST}(U_{\alpha}))\le\alpha\cdot w(\text{MST}(U))$.
	\end{proof}
	
	\begin{restatable}{remark}{SPRemark}
		The minimal possible stretch in the \Cref{thm:ultrametricStretchAlpha} above is $8$, which is obtained for lightness $\alpha=2$. This stretch is the best possible stretch obtained by a spanning tree. Indeed, consider the metric induced
		on the leaves of the full binary tree. One can observe that this is
		an ultrametric. Chan et al.~\cite{CXKR06} showed that for every $\epsilon>0$,
		there is a full binary tree large enough such that every tree over
		its set of leaves has stretch greater than $8-\epsilon$. 
		(see~\cite{Fil19,Fil20,FKT19,G01} for further details on the Steiner point removal problem.)
	\end{restatable}
	Note that a spanner with stretch smaller than $2$ might require $\Omega(n^{2})$
	edges. Indeed, the uniform metric (where all distances are $1$) is an
	ultrametric, and every spanner with a missing edge has stretch at
	least $2$. Similarly, it follows that every such spanner will have
	lightness $\Omega(n)$. In the next theorem we show that an online algorithm can get arbitrarily close to stretch $2$.
	
	\begin{restatable}{theorem}{SpannerUltrametric}
		\label{thm:SpannerUltrametric}
		Given an ultrametric $U$, for every $\epsilon\in(0,\frac{1}{2})$,
		an online algorithm can maintain an $(2+\eps)$-spanner with  $O(\epsilon^{-1}\log\eps^{-1})\cdot n$ edges 
		and $O(\eps^{-2})\cdot w(\MST)$ weight.
	\end{restatable}
	\begin{proof}
		For every $i\in\{0,1,\ldots,\kappa\}$ with 
		$\kappa=\left\lfloor \log_{1+\eps}\eps^{-1}\right\rfloor$,
		let $U_{i}$ be the ultrametric $U$, where for every pair of vertices,
		$d_{U_{i}}(u,v)$ is defined to be the $(1+\epsilon)^{i}\cdot\epsilon^{-j}$
		for the minimal index $j$ such that $d_{U_{i}}(u,v)\le(1+\epsilon)^{i}\cdot\epsilon^{-j}$.
		We construct a spanner $H_{i}$ for $U_{i}$ using the algorithm above.
		The final spanner will be $H=\bigcup_i H_{i}$ (with the original weights).
		
		The sparsity is straightforward, as we have $\kappa+1=O(\log_{1+\epsilon}\eps^{-1})=O(\epsilon^{-1}\log\eps^{-1})$
		trees. For the lightness, let $T$ be an MST of the ultrametric $X$.
		Denote by $T_{i}$ the MST of $U_{i}$. For every edge $e\in T$,
		it holds that 
		\[
		\sum_{i=0}^{\kappa}w_{U_{i}}(e)
		\le w_{U}(e)\sum_{i=0}^{\kappa}(1+\epsilon)^{i}
		=\frac{(1+\epsilon)^{\kappa+1}-1}{\epsilon}\cdot w_{U}(e)
		=O(\epsilon^{-2})\cdot w_{U}(e).
		\]
		We now can bound the weight of $H$ as follows:
		\begin{align*}
			w_{U}(H) 
			& \le\sum_{i=0}^{\kappa}w_{U}(H_{i})
			\le\sum_{i=0}^{\kappa}w_{U_{i}}(H_{i})
			=\sum_{i=0}^{\kappa}w_{U_{i}}(T_{i})
			\le\sum_{i=0}^{\kappa}w_{U_{i}}(T)\\
			& =\sum_{i=0}^{\kappa}\sum_{e\in T}w_{U_{i}}(e)
			=\sum_{e\in T}\sum_{i=0}^{\kappa}w_{U_{i}}(e)
			=\sum_{e\in T}O(\epsilon^{-2})\, w_{U}(e)
			=O(\epsilon^{-2})\, w_{U}(T).
		\end{align*}
		
		It remains to analyze the stretch of $H$. For every pair of vertices $u,v\in U$, there are unique indices $i,j$ such that $(1+\epsilon)^{i-1}\cdot\epsilon^{-j}<d_{U}(u,v)\le(1+\epsilon)^{i}\cdot\epsilon^{-j}$.
		Hence in $U_{i}$ it holds that $d_{U_{i}}(u,v)\le(1+\epsilon)\cdot d_{U}(u,v)$.
		As $U_{i}$ is an $\eps^{-1}$-HST, it holds that 
		\[
		d_{H}(u,v)
		\le d_{H_{i}}(u,v)
		\le\frac{2\eps^{-1}}{\eps^{-1}-1}\cdot d_{U_{i}}(u,v)
		\le\frac{2}{1-\epsilon}\cdot(1+\epsilon)\cdot d_{U}(u,v)
		\le 2 (1+3\epsilon)\cdot d_{U}(u,v)\,.
		\]
		One can obtain the stretch factor $2+\eps$, stated in the theorem, by scaling
		$\epsilon$ accordingly.
	\end{proof}

	\section{Conclusion}
	\label{sec:con}
	
	We studied online spanners for points in metric spaces. 
	In the Euclidean $d$-space, we presented an online $(1+\eps)$-spanner algorithm with competitive ratio $O(\eps^{1-d}\log n)$,
	improving the previous bound of $O_d(\eps^{-(d+1)}\log n)$ from~\cite{BT-oes-21}.
	In fact, the spanner maintained by the algorithm has $O_d(\eps^{1-d}\log \eps^{-1})\cdot n$ edges, almost matching the (offline) optimal bound of $O_d(\eps^{1-d})\cdot n$. 
	Moreover, in the plane, a tighter analysis of the same algorithm provides an almost quadratic improvement of the competitive ratio  to $O(\eps^{-3/2}\log\eps^{-1}\log n)$, by comparing the online spanner with an instance-optimal spanner directly, circumventing the comparison to an MST (i.e., lightness).
	Note that, the logarithmic dependence on $n$ is unavoidable due to a $\Omega((\eps^{-1}/\log \eps^{-1})\log n)$ lower bound in the real line~\cite{BT-oes-21}.
	However, our lower bound $\Omega(\eps^{-d})$ under $L_1$-norm in $\mathbb{R}^d$ shows a dependence on the dimension. 
	This leads to the following question.
	
	\begin{question*}\label{Q2}
		Does the competitive ratio of an online $(1+\eps)$-spanning algorithm for $n$ points in $\mathbb{R}^d$ necessarily grow proportionally with $\eps^{-f(d)}\cdot \log n$, where $\lim_{d\rightarrow \infty}f(d)=\infty$?
	\end{question*}
	
	Interestingly, for $t\in[(1+\epsilon)\sqrt{2},(1-\epsilon)2]$, we
	can show that every online $t$-spanner algorithm in $\mathbb{R}^{d}$
	must have competitive ratio $2^{\Omega(\epsilon^{2}d)}$ (see \Cref{thm:EuclideanLB} in \Cref{sec:EuclideanLB}).
	
	Next, we studied online spanners in general metrics. We showed that the \emph{ordered greedy} algorithm maintains a spanner with $O(\eps^{-1}\log\frac{k}{\eps}) \cdot n^{1+\frac{1}{k}}$ edges and $O(\eps^{-1}n^{\frac{1}{k}}\log^2 n)$ lightness, with stretch factor $t = (2k-1)(1+\eps)$ for $k\ge 2$ and $\eps\in(0,1)$, for a sequence of $n$ points in a metric space.
	Moreover, we show that these bounds cannot be significantly improved, by introducing an instance that achieves an $\Omega(\frac{1}{k}\cdot n^{1/k})$ competitive ratio on both sparsity and lightness.
	Finally, we established the trade-off among stretch, number of edges and lightness for points in ultrametrics, showing that one can maintain a $(2+\eps)$-spanner for ultrametrics with $O(n\cdot\eps^{-1}\log\eps^{-1})$ edges and $O(\eps^{-2})$ lightness.

	\newpage
	\bibliographystyle{alphaurlinit}
	\bibliography{arxiv}
	
	\appendix
	
	\section{High Dimensional Euclidean Lower Bound}\label{sec:EuclideanLB}
	\begin{theorem}\label{thm:EuclideanLB} 
		For $t\in[(1+\epsilon)\sqrt{2},(1-\epsilon)2]$, the competitive ratio of any online $t$-spanner algorithm in $\R^d$ under the Euclidean norm is $2^{\Omega(\epsilon^{2}d)}$.
	\end{theorem}
	\begin{proof}[Proof sketch.]
		Let $A\subseteq\left\{ \pm1\right\} ^{d}$ be a set of $2^{\Omega(\epsilon^{2}d)}$
		points such that every $u,v\in A$ differ in $(1\pm\epsilon)\frac{d}{2}$
		coordinates (such a set can be constructed randomly using Chernoff).
		In particular, $\|u-v\|_{2}$ is in $\sqrt{(1\pm\epsilon)2d}$. Every
		$t$-spanner for $A$ must contain all ${|A| \choose 2}$ edges, this
		is as the weight of any two edges is at least $2\sqrt{(1-\epsilon)2d}>t\cdot\sqrt{(1+\epsilon)2d}$. 
		
		Next the adversary introduces the point $\vec{0}$ with all zeros,
		which is at distance $\sqrt{d}$ from all other points. Let $H$ be
		the star with $\vec{0}$ as a center. Then for every $u,v\in A$,
		there is a path in $H$ of weight $2\sqrt{d}\le t\cdot\sqrt{(1-\epsilon)2d}\le t\cdot\|v-u\|_{2}$.
		The competitive ratio is $\Omega\left(\nicefrac{{A \choose 2}}{|A|}\right)=2^{\Omega(\epsilon^{2}d)}$.
	\end{proof}
	
\end{document}